\documentclass[11pt,a4paper]{article}
\usepackage[T1]{fontenc}
\usepackage{geometry}
\usepackage{amsmath}
\usepackage{amssymb}
\usepackage{amsthm}
\usepackage{mdframed}
\usepackage{graphicx}
\usepackage{bm}
\usepackage[inline]{enumitem}
\usepackage{tikz}
\usepackage{float}
\usetikzlibrary{matrix}
\usepackage{dsfont}
\usepackage[numbers]{natbib}
\usepackage[english]{babel}
\usepackage{caption}
\numberwithin{equation}{section}
\newcommand*{\lz}{\lambda}
\newcommand*{\id}{\mathop{}\!\mathrm{d}}
\newcommand*{\js}{\psi}
\newcommand*{\mjs}{\widehat{\js}}
\newcommand*{\bjs}{\widetilde{\js}}

\newcommand*{\jss}{\varphi}

\newcommand*{\st}{\sigma_3}

\newcommand{\determ}[2]{#1_{#2}^{\dagger} #1_{#2}}

\newcommand{\determz}[3]{#1_{#2}^{\dagger} {#3} #1_{#2} {#3}}

\newcommand*{\C}{\mathbb{C}}
\newcommand*{\LI}{L^1_{\scriptscriptstyle loc}(\R_+)}
\newcommand*{\R}{\mathbb{R}}
\newcommand*{\I}{\mathds{1}}

\DeclareMathOperator*{\Res}{Res}

\DeclareMathOperator{\sech}{sech}
\DeclareMathOperator{\diag}{diag}
\theoremstyle{plain}
\newtheorem{lem}{Lemma}[section]
\newtheorem{prop}[lem]{Proposition}
\newtheorem{cor}[lem]{Corollary}
\theoremstyle{definition}

\newtheorem{remark}[lem]{Remark}
\newmdtheoremenv{rhp}{Riemann--Hilbert problem}
\newcommand\Item[1][]{%
  \ifx\relax#1\relax  \item \else \item[#1] \fi
  \abovedisplayskip=0pt\abovedisplayshortskip=0pt~\vspace*{-\baselineskip}}
\begin{document}
\title{Soliton solutions of the nonlinear Schr\"odinger equation with defect conditions}
\author{K.T. Gruner%
\thanks{The author is supported by the SFB/TRR 191 `Symplectic Structures in Geometry,
Algebra and Dynamics', funded by the DFG.}\\\texttt{kgruner@math.uni-koeln.de}
\\\\ \emph{Universit{\"a}t zu K{\"o}ln, Mathematisches Institut,}
\\ \emph{Weyertal 86-90, 50931 K{\"o}ln, Germany}
}
\maketitle
  \begin{abstract}
  A recent development in the derivation of soliton solutions for initial-boundary
  value problems through Darboux transformations, motivated to reconsider solutions
  to the nonlinear Schr\"odinger (NLS) equation on two half-lines connected via
  integrable defect conditions. Thereby, the Darboux transformation to construct
  soliton solutions is applied, while preserving the spectral boundary constraint
  with a time-dependent defect matrix. In this particular model, \(N\)-soliton solutions
  vanishing at infinity are constructed. Further, it is proven that solitons are
  transmitted through the defect independently of one another.
  \end{abstract}
  {\bf Keywords:} NLS equation, integrable boundary conditions, star-graph,
  initial-boundary value problems, soliton solutions, dressing transformation, inverse scattering method.
\section{Introduction}
As an important physical equation the NLS equation was subject to a great number
of research works. Over time various methods to deal with integrable nonlinear PDEs in
different settings have been formulated. One of these methods, the Unified Transform,
announced in \cite{Fo} was successfully applied to initial-boundary value problems of
linear and integrable nonlinear PDEs of one space and one time variable. To this end,
the Unified Transform was used to yield results for the NLS equation regarding various
spatial domains like the half-line, a finite interval and even a star-graph \cite{Ca}.
As in the case for initial value problems, it is based on the representation of the
equation through a Lax pair which consists of two matrices usually referred to as the
\(t\) part and the \(x\) part. However, the structural innovation of the Unified Transform
is the simultaneous use of \(t\) and \(x\) part in the direct scattering process.

In some cases which mainly depend on the boundary condition, the Unified Transform is
for initial-boundary value problems as efficient as the inverse scattering transform
\cite{A} for initial value problems. These so-called linearizable
boundary conditions make use of a natural symmetry relation to linearize the problem on
the spectral side. Having identified linearizable boundary conditions, it is a priori not
clear that they are also integrable boundary conditions. Though, most of the known examples
conveniently fit both classes. Finding formulae for long-time asymptotics \cite{FoIt} and
for explicit solutions \cite{BiHw, Zh} for the NLS equation on the half-line
with certain and linearizable boundary conditions, respectively, has been well addressed
in the literature.

Nevertheless, the study of a defect or impurity at a fixed point which preserves integrability
is still of interest in rather recent studies by several authors,
not only for the NLS equation, but also for other PDEs. In one of these studies \cite{CoZa},
the authors illuminate the Lagrangian description of ``jump-defects'',
integrability preserving discontinuities with two fields \(u\), \(v\), where the conditions
relating the fields on the sides of the defect are B\"acklund transformations
frozen at the defect location. For the NLS equation on the two half-lines
they established the following defect at \(x=0\):
\begin{align*}
(u - v)_x &= i \alpha(u-v) +\Omega(u+v), \\
(u - v)_t &= -\alpha(u-v)_x +i\Omega(u+v)_x + i (u-v)(|u|^2+|v|^2),
\end{align*}
where \(\Omega = \sqrt{\beta^2 - |v-u|^2}\), \(\alpha\) and \(\beta\) real parameters
(\(\alpha\) was added in \cite{Ca2}).
Moreover, this jump-defect was also used by one of the authors to obtain new boundary condition
for the NLS equation on the half-line by combining them with Dirichlet boundary condition, see \cite{Za}.
It was shown that the defect condition \cite{Ca2} and also the new boundary condition
\cite{Za} have infinitely many conserved quantities and hence, they are integrable.
Moreover, the authors of \cite{CoZa} conjectured that in the model of the NLS equation with
defect conditions, an arbitrary number of solitons are transmitted through the defect
independently of one another. However, they have only proven this for particular cases
of one- and two-soliton solutions.

Using the aforementioned natural symmetry, a method called mirror-image technique
was developed to tackle initial-boundary value problems on the half-line by
extending it to the whole axis, which may seem like an unnatural approach.
On the other hand, there was recently a development for the Unified Transform \cite{Zh}
incorporating the Darboux transformation and hence the construction of exact solutions.
The method is, as it uses the Darboux transformation, highly reliant on the integrability
of the model. However given that it is, the idea of the method consists of the construction
of solutions while preserving the integrability. For the NLS equation with Robin boundary
conditions both methods were successfully applied, see \cite{BiHw} for the mirror-image
technique.

For integrable PDEs, the Darboux transformation \cite{GHZ, MaSa} is a
powerful method for constructing solutions. In particular, the well-known soliton
solution appearing in many physical motivated PDEs like the NLS equation can be
computed thereby. The crucial part of the new approach is to supplement the Darboux
transformation with the boundary conditions without destroying the integrability of
the system, which was realized in \cite{Zh} and called ``dressing the boundary''.

In this paper, our objective is to take up the described model of the NLS equation on two
half-lines together with defect conditions and compute exact solutions through the
dressing the boundary method, which already yielded results for a similar integrable model.
Then, only considering pure soliton solutions, we want
to prove the conjecture formulated in \cite{CoZa}, i.e.\@ each soliton
in the pure soliton solution is transmitted through the defect independently.
To the best knowledge of the author, combining boundary conditions corresponding
to a time-dependent boundary matrix with the latest method of computing exact solutions
of initial-boundary value problems \cite{Zh} extended to a star-graph is a novel
approach.

In Section \ref{sc:NLS}, we introduce the NLS equation and its equivalent spectral part
for which the inverse scattering transform is discussed. In particular, the analysis
for the Jost solutions and an understanding of the influence of parameter in the
construction of soliton solutions is crucial.
We present the methods of the B\"acklund transformations and Darboux transformations
in Section \ref{sc:back} and \ref{sc:darb}, respectively and discuss briefly the idea
of their connection. In preparation for dressing the boundary in the case of defect
conditions, we present their analogous spectral expression in Section~\ref{sc:def}
and prove some helpful properties.
Then, in Section~\ref{sc:wl} of this paper, we specify the model we want to solve:
the NLS equation on two half-lines connected via defect conditions at \(x=0\) and
realize the dressing the boundary in Proposition \ref{p:Nsol}.
Thereby, dressing the boundary lets us compute and visualize \(N\)-soliton solutions
in Section \ref{sc:vis}. Moreover, we discuss these solutions and prove that the solitons
are transmitted through the defect independently.
Finally, we gather further information and directions in the Conclusion. 
\section{Initial value problem for the NLS}\label{sc:NLS}
In the following, we give a brief summary of the inverse scattering transform of the
focusing NLS equation. As in \cite{BiHw} and \cite{Fu}, it will serve as a guideline
in order to implement additional results. Therefore, following the analysis given in
\cite{A}, we introduce the NLS equation
\begin{align}\label{eq:NLS}
\begin{split}
  &i u_t + u_{xx} + 2 |u|^2u = 0, \\
  &u(0,x) = u_0(x)
\end{split}
\end{align}
for \(u(t,x)\colon \R \times \R \mapsto \C\) and the initial condition \(u_0(x)\).
The equation can be expressed in an equivalent compatibility condition of the following
linear spectral problems
\begin{align}\label{eq:jost}
\begin{split}
  &\js_x = U\js,\\
  &\js_t = V\js,
\end{split}
\end{align}
where \(\js(t,x,\lz)\) and the matrix operators
\begin{equation}\label{eq:UV}
  U= -i \lz \st + Q, \quad
  V=-2 i\lz^2 \st +\widetilde{Q}
\end{equation}
are \(2\times2\) matrices. The potentials \(Q\) and \(\widetilde{Q}\)
of \(U\) and \(V\) are defined by
\begin{equation*}
  Q(t,x)=
  \begin{pmatrix}
    0 & u \\
    -u^\ast & 0
  \end{pmatrix},\quad
  \widetilde{Q}(t,x,\lz)=
  \begin{pmatrix}
    i|u|^2 & 2\lz u+iu_x \\
    -2\lz u^\ast+iu_x^\ast & -i|u|^2
  \end{pmatrix}\quad \text{and }
  \st=
  \begin{pmatrix}
    1 & 0 \\
    0 & -1
  \end{pmatrix}.
\end{equation*}
In this context, the matrices \(U\) and \(V\) form a so-called Lax pair,
depending not only on \(t\) and \(x\), but also on a spectral parameter \(\lz\).
Hereafter, the asterix denotes the complex conjugate, \(\C_+=\{\lz\in \C \colon
\Im(\lz)>0\}\) as well as \(\C_-=\{\lz\in \C \colon \Im(\lz)<0\}\)
and \(\js^\intercal\) is the transpose of \(\js\).
For a solution \(\js(t,x,\lz)\) of the Lax system~\eqref{eq:jost} the compatibility
condition \(\js_{tx} = \js_{xt}\) for all \(\lambda \in \C\)
is equivalent to \(u(t,x)\) satisfying the NLS equation \eqref{eq:NLS}. Moreover,
we will refer to \(U\) and \(V\) as the \(x\) and \(t\) part of the Lax pair,
respectively. In that regard, given a sufficiently fast decaying function
\(u(t,x)\to0\) and derivative \(u_x(t,x)\to0\) as \(|x|\to\infty\), it is reasonable
to assume that there exist \(2\times 2\)-matrix-valued solutions, we call modified Jost
solutions under time evolution, \(\mjs(t,x,\lz)=\js(t,x,\lz) e^{i\theta(t,x,\lz)\st}\),
where \(\theta(t,x,\lz)=\lz x+ 2 \lz^2 t\), of the modified Lax system
\begin{equation*}
  \mjs_x + i \lambda [\st,\mjs] = Q\mjs,\qquad
  \mjs_t + 2i\lambda^2[\st,\mjs]= \widetilde{Q}\mjs
\end{equation*}
with constant limits as \(x\to\pm\infty\) and for all \(\lz\in\R\),
\begin{equation*}
  \mjs_\pm(t,x,\lz) \to \I,  \qquad \text{as }x\to\pm\infty.
\end{equation*}
They are solutions to the following Volterra integral equations:
\begin{equation}\label{eq:Vie}
\begin{aligned}
  \mjs_-(t,x,\lz) &= \I + \int_{-\infty}^{x}
  e^{-i\theta(0,x-y,\lz)\st} Q(t,y) \mjs_-(t,y,\lz) e^{i\theta(0,x-y,\lz)\st}  \id y, \\
  \mjs_+(t,x,\lz) &= \I - \int_{x}^{\infty}
  e^{-i\theta(0,x-y,\lz)\st} Q(t,y) \mjs_+(t,y,\lz) e^{i\theta(0,x-y,\lz)\st}  \id y.
\end{aligned}
\end{equation}
\begin{lem}
Let \(u(t,\cdot)\in H^{1,1}(\R)=\{f\in L^2(\R)\colon xf,f_x\in L^2(\R)\}\).
Then, for every \(\lz\in\R\), there exist unique solutions \(\mjs_\pm(t,\cdot,\lz)
\in L^\infty(\R)\) satisfying the integral equations \eqref{eq:Vie}. Thereby,
the second column vector of \(\mjs_-(t,x,\lz)\) and the first column vector
of \(\mjs_+(t,x,\lz)\) can be continued analytically in \(\lz\in \C_-\) and
continuously in \(\lz\in \C_-\cup\R\), while the first column vector
of \(\mjs_-(t,x,\lz)\) and the second column vector of \(\mjs_+(t,x,\lz)\)
can be continued analytically in \(\lz\in \C_+\) and continuously in
\(\lz\in \C_+\cup\R\).
\end{lem}
Analogously, the columns of \(\js_\pm(t,x,\lz)\) can be continued analytically and
continuously into the complex \(\lz\)-plane, \(\js_-^{(2)}\) and \(\js_+^{(1)}\)
can be continued analytically in \(\lz\in \C_-\) and continuously in \(\lz\in \C_-\cup\R\),
while \(\js_-^{(1)}\) and \(\js_+^{(2)}\) can be continued analytically in
\(\lz\in \C_+\) and continuously in \(\lz\in \C_+\cup\R\).

The limits of the Jost solutions and the zero trace of the matrix \(U\) gives
\(\det \js_\pm = 1\) for all \(x\in\R\). Further, \( \js_\pm\) are both fundamental
matrix solutions to the Lax system \eqref{eq:jost}, so there exists an \(x\) and
\(t\) independent matrix \(A(\lz)\) such that
\begin{equation*}
   \js_-(t,x,\lz)=\js_+(t,x,\lz)A(\lz),\qquad \lz\in\R.
\end{equation*}
The scattering matrix \(A\) is determined by this system and therefore we can also
write \(A(\lz) = (\js_+(t,x,\lz))^{-1}\js_-(t,x,\lz)\), whereas its
entries can be written in terms of Wronskians. In particular,
\(a_{11}(\lz) = \det[\js_-^{(1)}|\js_+^{(2)}]\) and \(a_{22}(\lz) =
-\det[\js_-^{(2)}|\js_+^{(1)}]\) implying that they can respectively be
continued in \(\lz\in\C_+\) and \(\lz\in\C_-\). The eigenfunction inherit the
symmetry relation of the Lax pair
\begin{equation}\label{eq:jxsym}
  \js_\pm(t,x,\lz)= -\sigma \bigl(\js_\pm(t,x,\lz^\ast)\bigr)^\ast \sigma,
\end{equation}
which directly gives \(a_{22}(\lz) = a_{11}^\ast(\lz^\ast)\) and
\(a_{21}(\lz)=-a_{12}^\ast(\lz)\). The asymptotic behavior of the modified
Jost functions and scattering matrix as \(\lz\to\infty\) is
\begin{align*}
  \mjs_- & =\I+\frac{1}{2i\lz}\st Q+ \frac{1}{2i\lz}\st
  \int_{-\infty}^{x} |u(t,y)|^2 \id y+\mathcal{O}(1/\lz^2), \\
  \mjs_+ & =\I+\frac{1}{2i\lz}\st Q- \frac{1}{2i\lz}\st
  \int_{x}^{\infty} |u(t,y)|^2 \id y+\mathcal{O}(1/\lz^2)
\end{align*}
and \(A(\lz) = \I + \mathcal{O}(1/\lz)\).

Let \(u(t,\cdot)\in H^{1,1}(\R)\) be generic.
That is, \(a_{11}(\lz)\) is nonzero in \(\overline{\C_+}\) except at a finite
number of points \(\lz_1,\dots,\lz_N \in\C_+\), where it has simple zeros
\(a_{11}(\lz_j)=0\), \(a'_{11}(\lz_j)\neq0\), \(j=1,\dots,N\).
This set of generic functions \(u(t,\cdot)\) is an open dense subset of
\(H^{1,1}(\R)\) usually denoted by \(\mathcal{G}\). By the symmetry mentioned above,
\(a_{11}(\lz_j)=0\) if and only if \(a_{22}(\lz_j^\ast)=0\)
for all \(j=1,\dots, N\). At these zeros of \(a_{11}\) and \(a_{22}\),
we obtain for the Wronskians the following relation for \(j=1,\dots, N\),
\begin{equation}\label{eq:bj}
  \js_-^{(1)}(t,x,\lz_j)=b_j\js_+^{(2)}(t,x,\lz_j),\quad
  \js_-^{(2)}(t,x,\bar{\lz}_j)=\bar{b}_j\js_+^{(1)}(t,x,\bar{\lz}_j),
\end{equation}
where we defined \(\bar{\lz}_j=\lz_j^\ast\). Whereas for \(j=1,\dots,N\),
the relations then provide residue relations used in the inverse scattering method
\begin{equation*}
\begin{aligned}
  \Res_{\lz=\lz_j} \Bigl(\frac{\mjs_-^{(1)}}{a_{11}}\Bigr)
  &= C_j e^{2i\theta(t,x,\lz_j)}\mjs_+^{(2)}(t,x,\lz_j), \\
  \Res_{\lz=\bar{\lz}_j} \Bigl(\frac{\mjs_-^{(2)}}{a_{22}}\Bigr)
  &= \bar{C}_j e^{-2i\theta(t,x,\bar{\lz}_j)}\mjs_+^{(1)}(t,x,\bar{\lz}_j),
\end{aligned}
\end{equation*}
where the weights are \(C_j=b_j/a'_{11}(\lz_j)\) and \(\bar{C}_j=
\bar{b}_j/a'_{22}(\bar{\lz}_j)\), and they satisfy the symmetry relations
\(\bar{b}_j=-b_j^\ast\) and \(\bar{C}_j = -C_j^\ast\).

The inverse problem can be formulated using the jump matrix
\begin{equation*}
  J(t,x,\lz) =
  \begin{pmatrix}
    |\rho(\lz)|^2 & e^{-2i\theta(t,x,\lz)}\rho^\ast(\lz) \\
    e^{2i\theta(t,x,\lz)}\rho(\lz) & 0
  \end{pmatrix},
\end{equation*}
where the reflection coefficient is \(\rho(\lz) = a_{12}(\lz)/a_{11}(\lz)\)
for \(\lz\in\R\). Defining sectionally meromorphic functions
\begin{equation*}
  M_-=(\mjs_+^{(1)},\mjs_-^{(2)}/a_{22}),\qquad
  M_+=(\mjs_-^{(1)}/a_{11},\mjs_+^{(2)}),
\end{equation*}
we can give the method of recovering the solution \(u(t,x)\) from the scattering data.
\begin{rhp}\label{rhp}
For given scattering data \((\rho,\{\lz_j,C_j\}_{j=1}^N)\) as well as \(t,x\in \R\),
find a \(2\times2\)-matrix-valued function \(\C\setminus\R\ni\lz\mapsto
M(t,x,\lz)\) satisfying
\begin{enumerate}
  \item \(M(t,x,\cdot)\) is meromorphic in \(\C\setminus\R\).
  \item \(M(t,x,\lz) = 1 + \mathcal{O}(1/\lz)\) as \(|\lz|\to\infty\).
  \item Non-tangential boundary  values \(M_\pm(t,x,\lz)\) exist, satisfying the
  jump condition \(M_+(t,x,\lz)=M_-(t,x,\lz)(1+J(t,x,\lz))\) for
  \(\lz\in\R\).
  \item \(M(t,x,\lz)\) has simple poles at \(\lz_1,\dots,\lz_N,
      \bar{\lz}_1,\dots,\bar{\lz}_N\) with
  \begin{align*}
  \Res_{\lz=\lz_j} M(t,x,\lz) &=\lim_{\lz\to\lz_j} M(t,x,\lz)
  \begin{pmatrix}
    0 & 0 \\
    C_j e^{2i\theta(t,x,\lz_j)} & 0
  \end{pmatrix},\\
  \Res_{\lz=\bar{\lz}_j} M(t,x,\lz) &=\lim_{\lz\to\bar{\lz}_j}
  M(t,x,\lz)\begin{pmatrix}
    0 & \bar{C}_j e^{-2i\theta(t,x,\bar{\lz}_j)} \\
    0 & 0
  \end{pmatrix}.
  \end{align*}
\end{enumerate}
\end{rhp}
After regularization, the Riemann--Hilbert problem \ref{rhp} can be solved via
Cauchy projectors, and the asymptotic behavior of \(M_\pm(t,x,\lz)\) as \(\lz\to \infty\)
yields the reconstruction formula
\begin{align*}
  u(t,x) & = -2i\sum_{j=1}^{N} C_j^\ast e^{-2i\theta(t,x,\lz_j^\ast)}
  [\mjs_+^\ast]_{22}(t,x,\lz_j)\\
   & \quad - \frac{1}{\pi} \int_{-\infty}^{\infty}  e^{-2i\theta(t,x,\lz)}
   \rho^\ast(\lz) [\mjs_+^\ast]_{22}(t,x,\lz) \id\lz.
\end{align*}
In the reflectionless case, we have \(\rho(\lz)=0\) for \(\lz\in\R\) and
the Riemann--Hilbert problem can be reduced to an algebraic system
\begin{equation*}
\begin{aligned}
  \mjs_+^{(1)}(t,x,\lz_\ell) &= e_1 + \sum_{j=1}^{N}
  \frac{C_j e^{2i\theta(t,x,\lz_j)}\mjs_+^{(2)}(t,x,\lz_j)}{
  (\bar{\lz}_\ell-\lz_j)},\\
  \mjs_+^{(2)}(t,x,\lz_j) &= e_2 + \sum_{m=1}^{N}
  \frac{\bar{C}_m e^{-2i\theta(t,x,\bar{\lz}_j)}\mjs_+^{(1)}(t,x,\lz_m)}{
  (\lz_j-\bar{\lz}_m)}
\end{aligned}
\end{equation*}
for \(\ell,j=1,\dots,N\). The one-soliton solution is obtained for \(N=1\), we obtain
\begin{align*}
  [\mjs_+]_{21}(t,x,\lz_1) &= -\frac{C^\ast_1}{\lz_1-\lz_1^\ast}
  e^{-2i\theta(t,x,\lz_j^\ast)}\Biggl[1-\frac{|C_1|^2
  e^{2i(\theta(t,x,\lz_1)-\theta(t,x,\lz_1^\ast))}}{(\lz_1-\lz_1^\ast)^2}\Biggr]^{-1},\\
  [\mjs_+]_{22}(t,x,\lz_1) &= \Biggl[1-\frac{|C_1|^2
  e^{2i(\theta(t,x,\lz_1)-\theta(t,x,\lz_1^\ast))}}{(\lz_1-\lz_1^\ast)^2}\Biggr]^{-1}
\end{align*}
such that the one-soliton solution with \(\lz_1 = \xi+i\eta\) can be written as
\begin{equation*}
  u(t,x)=-2i\eta\frac{C_1^\ast}{|C_1|} e^{-i(2\xi x+4(\xi^2-\eta^2)t)}
  \sech\Bigl(2\eta(x+4\xi t)-\log\frac{|C_1|}{2\eta}\Bigr).
\end{equation*}
We change the notation so that \(u(t,x) = u_{1s}(t,x;\xi,\eta,x_1,\varphi_1)\) has the
following expression
\begin{equation}\label{eq:1sol}
  u_{1s}(t,x;\xi,\eta,x_1,\varphi_1)=2\eta
  e^{-i(2\xi x+4(\xi^2-\eta^2)t + (\varphi_1+\pi/2))}\sech(2\eta(x+4\xi t-x_1)),
\end{equation}
where \(\varphi_1= \arg(C_1)\) and \(x_1=\frac{1}{2\eta}\log\frac{|C_1|}{2\eta}\). 
\section{B\"acklund transformation}\label{sc:back}
Obtaining solutions for nonlinear partial differential equations is usually
not as easy as it may seem, given, we just constructed a one-soliton solution
for the NLS equation by the inverse scattering method. Apart from this method,
there is also the so-called B\"acklund transformation, which can be used to
obtain new solutions from a known solution by solving a system of integrable PDEs.
In the paper \cite{Ca2}, the author looked more generally at transformations
as the B\"acklund transformations and their implementation as boundary condition
at a given point on the line. In that regard, consider the Lax system
\eqref{eq:jost} for \(U\) and \(V\) as in \eqref{eq:UV}. By defining
\begin{equation*}
  \bjs(t,x,\lz)=B(t,x,\lz)\js(t,x,\lz),
\end{equation*}
we also consider the analog system
\begin{equation*}
  \begin{cases}
    \bjs_x = \widetilde{U}\bjs,\\
    \bjs_t = \widetilde{V}\bjs,
  \end{cases}
\end{equation*}
for \(\widetilde{U}\) and \(\widetilde{V}\) as in \eqref{eq:UV}
with \(u\) replaced by \(\tilde{u}\). Then, this definition gives us
partial differential equations for the so-called \emph{defect matrix} \(B\)
for any \(t\) and \(x\),
\begin{align}\label{eq:defect}
\begin{split}
  B_x&=\widetilde{U} B - B U,\\
  B_t&=\widetilde{V} B - B V.
\end{split}
\end{align}
Assuming that the defect matrix is linear in \(\lz\), one can show that
the matrix is of a particular form, see Proposition 2.2 in \cite{Ca2}.
\begin{prop}\label{p:def}
The defect matrix \(B=\lz B^{(1)}+ B^{(0)}\), relating Lax systems
corresponding to \(\tilde{u}\) and \(u\), has the following general
form in terms of NLS class equations
\begin{equation}\label{eq:genB}
  B(t,x,\lz)=2\lz\I +
  \begin{pmatrix}
    \alpha\pm i\sqrt{\beta^2-|\tilde{u} - u|^2} & -i(\tilde{u} - u) \\
    -i(\tilde{u} - u)^\ast & \alpha\mp i\sqrt{\beta^2-|\tilde{u} - u|^2}
  \end{pmatrix},
\end{equation}
where \(\alpha\in\R\), \(\beta\in \R\) are the parameter of the
defect and in particular, independent of \(t\) and \(x\).
\end{prop}
Here, it is important to note that the root
\(\sqrt{\beta^2-|\tilde{u}-u|^2}\) is real.
This fact follows from the symmetry \(U(t,x,\lz^\ast)^\ast =
\sigma U(t,x,\lz) \sigma^{-1}\) which transfers to \(B\), where
\begin{equation*}
  \sigma=\begin{pmatrix}
      0 & 1 \\
      -1 & 0
    \end{pmatrix}.
\end{equation*} 
\subsection{Darboux transformation}\label{sc:darb}
The structure of the Lax system permits the application of a, as we will see, very close
related class of transformations, the Darboux transformations. In this section,
we will outline the utilization of this method to directly obtain soliton solutions
of the NLS equation. Darboux transformations are known to provide an algebraic procedure
to derive soliton solutions of various integrable PDEs. In particular, they can be viewed
as gauge transformation acting on forms of the Lax pair \(U\), \(V\). Here,
it is meant to be applied while preserving certain constraints to
transform an ``old'' solution into a ``new'' solution. For that, the undressed Lax
system~\eqref{eq:jost} will be denoted as \(U[0]\), \(V[0]\) and \(\js[0]\) and
the transformed system as \(U[N]\), \(V[N]\) and \(\js[N]\),
whereby the solutions \(\js[0]\) and \(\js[N]\) are \(2\times1\) column solutions.
In terms of the B\"acklund transformation, \(U[0]\), \(V[0]\) and \(\js[0]\)
can be seen as the standard Lax system and \(U[N]\), \(V[N]\) and \(\js[N]\)
as the analog system with \(u\) replaced by \(\tilde{u}\).

Suppose that it is possible to construct a gauge-like transformation
\begin{equation*}
\js[1] = D[1] \js[0]
\end{equation*}
such that the structure of matrices
\begin{align}\label{eq:D1}
\begin{split}
  &U[1]=(D[1]_x +D[1]U[0])D[1]^{-1},\\
  &V[1]=(D[1]_t +D[1]V[0])D[1]^{-1}
\end{split}
\end{align}
is identical with the structure of \(U[0]\), \(V[0]\), i.e.\@ \(Q[0]\) becomes \(Q[1]\) with
updated off-diagonal entries. Indeed, if \(U[1]\) and \(V[1]\) satisfy \eqref{eq:D1},
then the undressed Lax system~\eqref{eq:jost} can be transformed into
\begin{align*}
\begin{split}
  &\js[1]_x = U[1]\js[1],\\
  &\js[1]_t = V[1]\js[1].
\end{split}
\end{align*}
At this point, it seems that a pair of solutions \(\js[0]\) and \(\js[1]\) is needed to determine
\(D[1]\). However, if and only if we are able to compute \(D[1]\) solely by a solution \(\js[0]\) of
the undressed Lax system \eqref{eq:jost}, we can construct new solutions and
then we call \(D[1]\) dressing matrix. Indeed, given a column solution
\(\js_1 = (\mu_1,\nu_1)^\intercal\) of the undressed Lax system at \(\lambda = \lambda_1\),
we write \(D[1]\) in the following form, which satisfies the requirement,
\begin{equation*}
D[1] = (\lambda-\lambda^\ast_1)\I+(\lambda^\ast_1-\lambda_1) P[1],\quad
P[1] = \frac{\js_1 \js_1^\dagger}{\determ{\js}{1}},
\end{equation*}
where \(\I\) is the identity and \(P[1]\) is a projector matrix.
Here, \(\js_1^\dagger\) denotes the transpose complex conjugate of \(\js_1\).
The important point of this method is that the solution \(u[1]\) can be
reconstructed through the first line of \eqref{eq:D1} or in terms of matrices
\begin{equation*}
Q[1]=Q[0] - i (\lambda_1-\lambda^\ast_1)[\st,P[1]],
\end{equation*}
which is called reconstruction formula.
Technically, the Darboux transformation can be summarized in the following way:
Suppose we have a system, of which we know the solution. Then, transforming the system via
the dressing matrix allows to construct the solution to a \emph{different} system.

Especially, if both systems correspond to the same PDE, the reconstruction formula lets us
obtain a new solution of the PDE.
Therefore, in advance a good understanding of the set of solutions of the NLS equation is instrumental,
since they are decisive when it comes to solutions of the Lax system.
However, there is only a limited number of significant cases known,
e.g.\@ the zero solution. In this regard, using the zero solution as seed solution,
i.e.\@ \(u[0]=0\), one can construct among other solutions a one-soliton solution \(u[1]\),
see \eqref{eq:1sol}. This will be of interest in the following studies.

Given \(N\) linear independent column solutions \(\js_j = (\mu_j,\nu_j)^\intercal\)
of the undressed Lax system \eqref{eq:jost} evaluated at \(\lambda=\lambda_j\),
\(j=1\dots N\), the basic dressing matrix \(D[1]\) may be iterated in the following sense
\begin{equation*}
D[N] = ((\lambda-\lambda^\ast_N)\I+(\lambda^\ast_N-\lambda_N) P[N]) \cdots
((\lambda-\lambda^\ast_1)\I+(\lambda^\ast_1-\lambda_1) P[1]),
\end{equation*}
where \(P[j]\) are projector matrices defined by
\begin{equation}\label{eq:Pj}
P[j] = \frac{\js_j[j-1] \js^\dagger_j[j-1]}{\determz{\js}{j}{[j-1]}},
\quad \js_j[j-1]=D[j-1]\big\rvert_{\lambda=\lambda_j} \js_j.
\end{equation}
Note that it is sufficient for the \(\lambda=\lambda_j\), \(j=1\dots N\), to be distinct
in order for the solutions to be linearly independent. Analogously to \(N=1\),
for the reconstruction formula we need to insert \(\js[N] = D[N] \js[0]\)
into the transformed Lax system
\begin{align*}
\begin{split}
  &\js[N]_x = U[N]\js[N],\\
  &\js[N]_t = V[N]\js[N],
\end{split}
\end{align*}
and extract the information of the coefficient of \(\lambda^{N-1}\) of the first line. 
Then, the reconstruction formula can be computed as
\begin{equation}\label{eq:recon}
Q[N]=Q[0] - i \sum_{j=1}^{N} (\lambda_j-\lambda^\ast_j)[\st,P[j]].
\end{equation}
In the course of this paper we will also work with the half-line as domain,
for which this construction can be done in the same way, resulting in a solution
\(u[N]\) on the half-line given through \eqref{eq:recon}.

Since the Darboux transformation is a matrix transforming undressed Lax systems
of the NLS equation to Lax systems of the NLS equation, there is a correspondence
between matrices from the Darboux transformation and of defect form. We will address
this idea in the next remark.
\begin{remark}\label{r:defdar}
\begin{enumerate}[label=(\roman*)]
\item The one-fold dressing matrix \(D[1]\), constructed by \(\C\setminus\R\ni\lz_1 =
\xi+ i\eta\) and \(\js_1\in \C^2\), satisfies \eqref{eq:defect}
with \(\widetilde{U}=U[1]\) and \(U=U[0]\). Given \emph{certain} information,
we can, up to a function of \(\lz\), write \(D[1]\) in the form of a defect matrix
with \(\gamma\in \R\), \(\delta\in \R\setminus\{0\}\), that is
\begin{equation}\label{eq:darb1}
  D[1] = \lz \I + \frac{1}{2}
  \begin{pmatrix}
    -2\gamma \pm i\sqrt{4\delta^2-|u[1]-u[0]|^2}& -i(u[1]-u[0]) \\
    -i(u[1]-u[0])^\ast & -2\gamma \mp i\sqrt{4\delta^2-|u[1]-u[0]|^2}
  \end{pmatrix}.
\end{equation}
\item The defect matrix \(B= 2\lz\I+ B^{(0)}\) in general form
\eqref{eq:genG} with \(\alpha, \beta \in \R\), applied in the context of
NLS class equations, admits a projector matrix \(P_0\) if \(\beta\neq0\).
Take \(\lz_0=-\frac{\alpha- i\beta}{2}\) and define
\begin{equation*}
  P_0=\frac{1}{2(\lz_0^\ast-\lz_0)} (B^{(0)}+\lz_0^\ast \I).
\end{equation*}
Then, we can write \(B\), up to a function of \(\lz\), in the form of a dressing matrix
\begin{equation*}
  B = 2(\lz-\lz_0^\ast)\I + 2(\lz_0^\ast-\lz_0) P_0.
\end{equation*}
\end{enumerate}
\end{remark}
This remark states that: if \(\beta\neq0\) (or \(\eta\neq0\)),
the general form of a defect matrix and a one-fold dressing matrix
are in certain cases interchangeable. In particular, this means that
for a defect matrix \(B(t,x,\lz)\) with a parameter \(\beta\neq0\), we know that
there exists a vector \(\upsilon_0\) which is in the kernel of \(B(t,x,\lz_0)\).
Conversely, taking a matrix polynomial of order one with a kernel vector
\(\upsilon_0\) at \(\lz=\lz_0\) corresponding to a one-fold dressing matrix,
i.e.\@ it solves \eqref{eq:D1}, we might be able to write it in the form of
a defect matrix \eqref{eq:defect}. However, it should be noted that coming
from a dressing matrix \(D[1]\), it is also a priori not clear what the
corresponding sign in front of the root in the \((11)\)-entry and accordingly
the \((22)\)-entry has to be. In some cases, as we will see, this information can be
extracted from the kernel vectors of the dressing matrix. On the other hand, coming
from a defect matrix \(B\), we only have the projector matrix in terms of expressions
of the solution side without knowledge of how the kernel vectors look like in terms of
the spectral side. Again, in some cases, it is possible to obtain information on the
kernel vectors from the signs in front of the root. 
\subsection{Localized B\"acklund transformation}\label{sc:def}
The B\"acklund transformation has also been investigated as frozen at a specific
point \(x_{f}\) and with that in mind as a means to generate integrable boundary
value systems. We will introduce the idea of this method in this section.
Restricting \(\tilde{u}\) and \(u\) to solutions of the NLS equation on different
half-lines and therefore also their Lax systems respectively to
\((t,x)\in\R_+\times\R_-\) and \((t,x)\in\R_+\times\R_+\),
we simultaneously restrict \eqref{eq:defect} to \(x_{f}=0\) on the line.
To distinguish between the B\"acklund transformation and the localized
B\"acklund transformation, we denote the denote the defect matrix by
\(B(t,x,\lz)\) as before and the localized defect matrix by \(G(t,0,\lz)\).
For any \(t\in\R_+\) and \(x=0\), we call the relations
\begin{align}\label{eq:locdef}
\begin{split}
  G_x&=\widetilde{U} G - G U,\\
  G_t&=\widetilde{V} G - G V,
\end{split}
\end{align}
\emph{boundary constraint}. From Proposition \ref{p:def} it follows that
assuming \(G(t,0,\lz)\) is linear in \(\lz\), it admits the form
\begin{equation}\label{eq:genG}
  G(t,0,\lz)=2\lz\I +
  \begin{pmatrix}
    \alpha\pm i\sqrt{\beta^2-|\tilde{u} - u|^2} & -i(\tilde{u} - u) \\
    -i(\tilde{u} - u)^\ast & \alpha\mp i\sqrt{\beta^2-|\tilde{u} - u|^2}
  \end{pmatrix}
\end{equation}
with \(\alpha, \beta \in \R\).
In \cite{Ca2}, the defect matrix \(B(t,x,\lz)\) has been discussed
not only for equations of NLS type for \(u\) and \(\tilde{u}\), but as an
universal approach to Lax pair systems. Hence, different defect matrices
and consequently different localized defect matrices could be identified
corresponding to the Lax systems of various PDEs.


Also note, that \eqref{eq:locdef} at \(x=0\) has a structural difference to the boundary
constraint for the half-line, see \cite{Zh}. For the model on two half-lines, we
are not relating \(V(t,0,\lz)\) and \(V(t,0,-\lz)\) as for the half-line,
but \(V(t,0,\lz)\) and \(\widetilde{V}(t,0,\lz)\). Furthermore, the relation
of \(V(t,0,\lz)\) and \(\widetilde{V}(t,0,\lz)\) immediately implies the
relation for \(U(t,0,\lz)\) and \(\widetilde{U}(t,0,\lz)\).

Let us discuss, which information is needed in order to determine the sign in front of
the root in the \((11)\)-entry, considering we already constructed the localized
defect form from the one-fold dressing matrix at \(x=0\). In that regard, important
properties of a B\"acklund transformation with respect to \(x\) has been in detail discussed
in the paper \cite{DZ}. In particular, it was shown that the transformation
\(\mathcal{B}^+_{\Im(\lz_1)}(\js_1) \colon u \mapsto \tilde{u}
= \mathcal{B}^+_{\Im(\lz_1)}(\js_1) u\), the B\"acklund transformation of \(u(t,\cdot)\)
with respect to \(\{\Im(\lz_1),\js_1\}\) on \(\R_+\), is a bijection
from \(H^{1,1}(\R_+)\) onto \(H^{1,1}(\R_+)\). Similarly, we want to analyze
the localized defect matrix as B\"acklund transformation with \(\beta\neq0\)
and with respect to \(t\). For functions \(f(\cdot,0,\lz)\), we introduce the
function spaces
\begin{align*}
&H^{0,1}_t(\R_+) = \{f\in L^2(\R_+)\colon t f \in L^2(\R_+)\},\\
&H^{1,1}_t(\R_+) = \{f\in L^2(\R_+)\colon \partial_t f, t f \in L^2(\R_+)\}
\end{align*}
and state the following lemma, which will be essential in the proof.
\begin{lem}\label{l:est}
Let \(f(\cdot,0,\lz)\in H^{0,1}_t(\R_+)\), \(g(\cdot,0,\lz)\in H^{1,1}_t(\R_+)\) and
\(\Im(\lz^2)< 0\). Then,
\begin{align*}
  \Bigl\|\int_{\langle t\rangle}^{\infty}f(\tau,0,\lz) g(\tau,0,\lz)\id\tau
  \Bigr\|_{H^{1,1}_t(\R_+)}
  & \leq c \|f(\cdot,0,\lz)\|_{H^{0,1}_t(\R_+)}
  \|g(\cdot,0,\lz)\|_{H^{1,1}_t(\R_+)}, \\
  \Bigl\|\int_{\langle t\rangle}^{\infty}f(\tau,0,\lz) e^{-4\Im(\lz^2)
  (\langle t\rangle-\tau)}\id\tau \Bigr\|_{H^{1,1}_t(\R_+)}
  & \leq c \|f(\cdot,0,\lz)\|_{H^{0,1}_t(\R_+)},
\end{align*}
where \(c\) depends on \(\lz\).
\end{lem}
\begin{proof}
Analogously to the proof in \cite{DZ}.
\end{proof}
We skip the part in \cite{DZ}, where the defect matrix is shown to be of the form
of a dressing matrix and immediately assume we are given a spectral parameter
\(\lz_1=\xi+i\eta\) together with a vector \(\js_1=(\mu,\nu)^\intercal\), which is
a solution of the Lax system \eqref{eq:jost} at \(\lz=\lz_1\), from which
we can construct a dressing matrix \(D[1]\) such that
\begin{equation*}
  D[1] = \lz \I+
  \begin{pmatrix}
    -\xi-i\eta\frac{|\mu|^2-|\nu|^2}{|\mu|^2+|\nu|^2} &
    \frac{-2i\eta\mu\nu^\ast}{|\mu|^2+|\nu|^2} \\
    \frac{-2i\eta\mu^\ast\nu}{|\mu|^2+|\nu|^2} &
    -\xi+i\eta\frac{|\mu|^2-|\nu|^2}{|\mu|^2+|\nu|^2}
  \end{pmatrix}.
\end{equation*}
Fixing \(x=0\), we can see the dressing matrix as a connection of two NLS equations
on the respective half-lines \(\R_-\) and \(\R_+\) and hence, it satisfies the
boundary constraint \eqref{eq:locdef} for some solutions \(u_1(t,0)\) and \(u(t,0)\)
of the NLS equations, whereas
\begin{equation}\label{eq:u1}
\begin{aligned}
  u_1(t,0) &= u(t,0) -2\eta \frac{\mu\nu^\ast}{|\mu|^2+|\nu|^2},\\
  (u_1)_x(t,0) &= u_x(t,0) -4 u(t,0)
  i\eta\frac{|\mu|^2-|\nu|^2}{|\mu|^2+|\nu|^2}
  -2\eta\frac{\mu\nu^\ast}{|\mu|^2+|\nu|^2}
  \Bigl(-\xi+i\eta\frac{|\mu|^2-|\nu|^2}{|\mu|^2+|\nu|^2}\Bigr),
\end{aligned}
\end{equation}
where these can be derived by \eqref{eq:D1}.
So that we have that a transformation \(\mathcal{B}^t_{\lz_1}(\js_1)
\colon u\mapsto u_1=\mathcal{B}^t_{\lz_1}(\js_1)u\) mapping
\(u(\cdot,0)\in\LI \to \LI\ni u_1(\cdot,0)\).
The denominator \(|\mu|^2+|\nu|^2\) can not be zero, since \(\js_1\)
is a solution of \(\js_t = (-2i\lz^2\st+\widetilde{Q})\js\) at \(\lz=\lz_1\).
If there exists a \(t_0\in\R_+\) such that \(\js_1=0\), then \((\js_1)_t=0\)
at \(t_0\in\R_+\) and thereby \(\js_1=0\) for every \(t\in\R_+\).
Assuming a nonzero asymptotic limit of \(\js_1\) gives the contradiction.
In particular, the transformation has a left inverse. Take
\(\mathcal{B}^t_{\lz_2}(\js_2)\), where \(\lz_2=\lz_1\),
\(\js_2=(-k_1\nu^\ast,k_1\mu^\ast)^\intercal\), \(k_1\in\C\), we have
\begin{equation*}
  u_2(t,0) = u_1(t,0) +2\eta \frac{\mu\nu^\ast}{|\mu|^2+|\nu|^2}=u(t,0).
\end{equation*}
We define
\begin{equation*}
  X=\{f\in  H^{1,1}_t(\R_+), \, f_x\in  H^{0,1}_t(\R_+)\}
\end{equation*}
with \(u(\cdot,0,\lz) \in \{f\in  H^{0,1}_t(\R_+),
\, f_x\in  H^{0,1}_t(\R_+)\} \subset X\), we can show that
\begin{align*}
  \|\widetilde{Q}(\cdot,0,\lz)\|_{L^1(\R_+)} &\leq \begin{pmatrix}
    \|u^2\|_{L^1(\R_+)} & \|2\lz u + i u_x\|_{L^1(\R_+)} \\
    \|2\lz u^\ast + i u_x^\ast\|_{L^1(\R_+)} & \|u^2\|_{L^1(\R_+)}
  \end{pmatrix}\\
  &\leq \begin{pmatrix}
    \|u\|^2_{L^2(\R_+)} & 2|\lz|\|u\|_{H^{0,1}_t(\R_+)}+
    \|u_x\|_{H^{0,1}_t(\R_+)} \\
    2|\lz|\|u\|_{H^{0,1}_t(\R_+)} + \|u_x\|_{H^{0,1}_t(\R_+)} &
    \|u\|^2_{L^2(\R_+)}
  \end{pmatrix}.
\end{align*}
So that each component is bounded by a constant depending on \(c(\lz)\) multiplied by
\(\|u\|_X\). Then, we can prove that
\begin{prop}\label{p:bij}
\(\mathcal{B}^t_{\lz_1}(\js_1)\), where \(\lz_1\in\C\setminus(\R\cup i\R)\),
maps functions \(u(\cdot,0) \in X\) onto \(\tilde{u}(\cdot,0) \in X\).
\end{prop}

\begin{proof}
Following the proof for the B\"acklund transformation with respect to \(x\),
see Proposition 4.7 in \cite{DZ}, we want to introduce a \(t\) dependent (Jost) function.
In that regard, we freeze the space variable \(x\), whereas we need it particularly
at \(x=0\). Then, given the limit behaviors \(|u(t,0)|\to 0\) and \(|u_x(t,0)|\to 0\) as
\(t \to \infty\), it is reasonable to assume that there exists a
\(2\times1\)-vector-valued solution \(m\) to the spectral problem
\begin{equation*}
  \js_t = (-2i\lz^2\st + \widetilde{Q})\js
\end{equation*}
admitting the asymptotic behavior \(m(t,0,\lz) \sim e_1 e^{-2i\lz^2 t}\)
as \(t\to\infty\). Then, we also define the normalized \(t\) dependent (Jost) function by
\begin{equation*}
  \hat{m}(t,0,\lz) = m(t,0,\lz) e^{2i\lz^2 t},
\end{equation*}
whereas it admits the normalization \(\lim_{t\to\infty} \hat{m}(t,0,\lz) = e_1\).
The solution \(m(t,0,\lz) = \hat{m}(t,0,\lz)e^{-2i\lz^2 t}\) is uniquely specified
by the asymptotic behavior \(\hat{m}(t,0,\lz)\to e_1\) as \(t\to\infty\).
The normalized (Jost) function is a solution to the following Volterra integral equation
\begin{equation}\label{eq:volmh}
  \hat{m}(t,0,\lz)=e_1-\int_{t}^{\infty}
  \begin{pmatrix}
    1 & 0 \\
    0 & e^{4i\lz^2 (t-\tau)}
  \end{pmatrix}
  \widetilde{Q}(\tau,0,\lz)\, \hat{m}(\tau,0,\lz) \id\tau.
\end{equation}
This, we will show by defining the operator
\begin{equation*}
  \mathcal{M}[\hat{m}](t,0,\lz)=-\int_{t}^{\infty}
  \begin{pmatrix}
    1 & 0 \\
    0 & e^{4i\lz^2 (t-\tau)}
  \end{pmatrix}
  \widetilde{Q}(\tau,0,\lz)\, \hat{m}(\tau,\lz) \id\tau,
\end{equation*}
which is a bounded operator mapping from \(L^\infty(\R_+)\) to \(L^\infty(\R_+)\)
for any fixed \(\lz\) such that \(\Im(\lz^2)<0\), since \(t-\tau\leq0\).
Also, we define
\begin{equation*}
  \mathcal{M}_j[\hat{m}](t,0,\lz)=-\int_{t}^{t_{j-1}}
  \begin{pmatrix}
    1 & 0 \\
    0 & e^{4i\lz^2 (t-\tau)}
  \end{pmatrix}
  \widetilde{Q}(\tau,0,\lz)\, \hat{m}(\tau,0,\lz) \id\tau,
\end{equation*}
where we fix \(\lz\) such that \(\Im(\lz^2)=0\). For an arbitrary interval
\((t_{j-1},t_j) \subset \R_+\), we obtain the estimate
\begin{equation*}
  ||\mathcal{M}_j[\hat{m}](\cdot,0,\lz)||_{L^\infty(t_{j-1},t_j)}\leq
  ||\widetilde{Q}(\cdot,0,\lz)||_{_{L^1(t_{j-1},t_j)}}
  ||\hat{m}(\cdot,0,\lz)||_{L^\infty(t_{j-1},t_j)}.
\end{equation*}
Then, we can choose \(t_j\) in such a way that the operator \(\mathcal{M}_j\) is a
contraction from \(L^\infty(t_{j-1},t_j)\) to \(L^\infty(t_{j-1},t_j)\). Repeating
this argument starting from \(t_0=0\) and appropriately chosen \(t_1\), \dots,
to \(t_{\ell-1}\) and \(t_\ell = \infty\), we can obtain finitely many intervals
such that \(\mathcal{M}_j\) is contraction from \(L^\infty(t_{j-1},t_j)\) to
\(L^\infty(t_{j-1},t_j)\), \(j=1,\dots,\ell\). Setting \(\hat{m}_0(t,0,\lz)\equiv e_1\)
on \((t_0,t_1)\), we can find a function \(\hat{m}_j(\cdot,0,\lz)\in
L^\infty(t_{j-1},t_j)\) by the Banach Fixed Point Theorem such that
it solves the equation
\begin{equation*}
    \hat{m}_j(t,0,\lz)= \hat{m}_{j-1}(t_j,0,\lz)+
    \mathcal{M}_j[\hat{m}_j](t,0,\lz), \qquad t\in(t_{j-1},t_j)
\end{equation*}
for every \(j=2,\dots,\ell\). Combining these functions, we find a continuous function
in \(L^\infty(\R_+)\) satisfying the Volterra integral equation \eqref{eq:volmh},
which covers the existence of \(\hat{m}(t,0,\lz)\). Having two solutions
\(\hat{m}(t,0,\lz)\) and \(\tilde{m}(t,0,\lz)\) to the Volterra integral equation
\eqref{eq:volmh}, we can deduce
\begin{equation*}
  |\hat{m}(t,0,\lz)-\tilde{m}(t,0,\lz)|\leq
  \int_{t}^{\infty} |\widetilde{Q}(\tau,0,\lz)|\, |\tilde{m}(\tau,0,\lz)- \hat{m}(\tau,0,\lz)| \id\tau.
\end{equation*}
Then by Gr{\"o}nwall's lemma, we obtain uniqueness.

Now, for the claims regarding the continuation of \(\hat{m}(t,0,\lz)\)
to \(\Im(\lz^2)\leq0\). Analogously to the \(x\) dependent Jost solution
\(\mjs_-^{(1)}(t,x,\lz)\), we introduce for \(\hat{m}(t,0,\lz)\)
the Neumann series \(\sum_{j=0}^{\infty}
\mathcal{M}^j[m_0](t,0,\lz)\), where \(m_0(t,0,\lz)\equiv e_1\), which is
formally a  solution to the Volterra integral equation \eqref{eq:volmh}.
Then, it is possible to derive a bound of the iterated operator \(\mathcal{M}\).
We define \(T(t,\lz)\) by
\begin{equation*}
  T(t,\lz)=\int_{t}^{\infty}|\widetilde{Q}(\tau,0,\lz)|\id\tau
  \leq\int_{0}^{\infty}|\widetilde{Q}(\tau,0,\lz)|\id\tau\leq
  \|\widetilde{Q}(\cdot,0,\lz)\|_{L^1(\R_+)}.
\end{equation*}
And then show that by induction, we have
\begin{align*}
  |\mathcal{M}^{j+1}[\hat{m}](t,0,\lz)| & \leq
  c \frac{\|\hat{m}(\cdot,0,\lz)\|_{L^\infty(\R_+)}}{j!}
  \int_{t}^{\infty} |\widetilde{Q}(\tau,0,\lz)|\, (T(\tau,\lz))^j \id\tau\\
  &\leq c \frac{\|\hat{m}(\cdot,0,\lz)\|_{L^\infty(\R_+)}}{j!}
  \int_{0}^{T(t,\lz)} s^j \id s\\
  &= c \|\hat{m}(\cdot,0,\lz)\|_{L^\infty(\R_+)}
  \frac{(T(t,\lz))^{j+1}}{(j+1)!},
\end{align*}
where we put \(s=T(\tau,\lz)\). Thus, we have that \(\sum_{j=0}^{\infty}
\mathcal{M}^j[m_0](t,0,\lz)\) is majorized in norm by a uniformly convergent
power series and is therefore itself uniformly convergent for \(\Im(\lz)\leq0\).
The analyticity and continuity continuation for \(\hat{m}(t,0,\lz)\) respectively
in \(\{\lz\in\C\setminus\{0\}\colon \Im(\lz^2)\leq0\}\) and in
\(\{\lz\in\C\setminus\{0\}\colon \Im(\lz^2)<0\}\) holds also, as before,
for the function \(m(t,0,\lz)\). It is left, to show that the entries of
\(\hat{m}(\cdot,0,\lz)-e_1\) are in \(H^{1,1}_t(\R_+)\). Since \(\mathcal{M}\) maps
\(L^\infty(\R_+)\) to \(L^\infty(\R_+)\) and writing \(\hat{m}(t,0,\lz)=
(\hat{m}_1, \hat{m}_2)\), we can estimate using Lemma \ref{l:est},
\begin{align*}
  \|\hat{m}_2(\cdot,0,\lz)\|_{H^{1,1}_t(\R_+)}&\leq
  c\|(\widetilde{Q}_{21}\hat{m}_1)(\cdot,0,\lz)\|_{H^{0,1}_t(\R_+)}+
  c\|(\widetilde{Q}_{22}\hat{m}_2)(\cdot,0,\lz)\|_{H^{0,1}_t(\R_+)}\\
  &\leq\|\hat{m}_1(\cdot,0,\lz)\|_{L^\infty(\R_+)}
  \|\widetilde{Q}_{21}(\cdot,0,\lz)\|_{H^{0,1}_t(\R_+)}\\&\quad
  +\|\hat{m}_2(\cdot,0,\lz)\|_{L^\infty(\R_+)}
  \|\widetilde{Q}_{22}(\cdot,0,\lz)\|_{H^{0,1}_t(\R_+)}
\end{align*}
and
\begin{align*}
  \|\hat{m}_1(\cdot,0,\lz)-1\|_{H^{1,1}_t(\R_+)}&\leq
  c\|(\widetilde{Q}_{11}\hat{m}_1)(\cdot,0,\lz)\|_{H^{1,1}_t(\R_+)}+
  c\|(\widetilde{Q}_{12}\hat{m}_2)(\cdot,0,\lz)\|_{H^{1,1}_t(\R_+)}\\
  &\leq\|\hat{m}_1(\cdot,0,\lz)\|_{L^\infty(\R_+)}
  \|\widetilde{Q}_{11}(\cdot,0,\lz)\|_{H^{1,1}_t(\R_+)}\\&\quad
  +\|\hat{m}_2(\cdot,0,\lz)\|_{H^{1,1}_t(\R_+)}
  \|\widetilde{Q}_{12}(\cdot,0,\lz)\|_{H^{0,1}_t(\R_+)}.
\end{align*}
And for the entries of \(\widetilde{Q}(t,0,\lz)\), we find
\begin{equation}\label{eq:wQest}
\begin{aligned}
  \|\widetilde{Q}_{11}(\cdot,0,\lz)\|_{H^{1,1}_t(\R_+)}&\leq
  \|u(\cdot,0)\|_{L^\infty(\R_+)}\|u(\cdot,0)\|_{H^{1,1}_t(\R_+)}, \\
  \|\widetilde{Q}_{12}(\cdot,0,\lz)\|_{H^{0,1}_t(\R_+)}&\leq
  2|\lz|\|u(\cdot,0)\|_{H^{0,1}_t(\R_+)}+\|u_x(\cdot,0)\|_{H^{0,1}_t(\R_+)}, \\
  \|\widetilde{Q}_{21}(\cdot,0,\lz)\|_{H^{0,1}_t(\R_+)}&\leq
  2|\lz|\|u(\cdot,0)\|_{H^{0,1}_t(\R_+)}+\|u_x(\cdot,0)\|_{H^{0,1}_t(\R_+)}, \\
  \|\widetilde{Q}_{22}(\cdot,0,\lz)\|_{H^{0,1}_t(\R_+)}&\leq
  \|u(\cdot,0)\|_{L^\infty(\R_+)}\|u(\cdot,0)\|_{H^{0,1}_t(\R_+)}.
\end{aligned}
\end{equation}
Thereby, if \(u(\cdot,0)\in X\), then \(\hat{m}(\cdot,0,\lz)-e_1\in H^{1,1}_t(\R_+)\).
Next, we consider a solution \(n(t,0,\lz)\) of the \(t\) part of the Lax pair
defined on \(\Im(\lz^2)\leq0\) and \(t\in\R_+\) with the property
\begin{align*}
  n(t,0,\lz) &= (e_2 + r_1(t))e^{2i\lz^2 t},
  \quad r_1 \in H^{1,1}_t(\R_+).
\end{align*}
Here, \(e_2=(0,1)^\intercal\) and \(n(t,0,\lz)\) is a non-unique solution to
the differential equation defined on the same domain as \(m(t,0,\lz)\) and,
as we will show, is linear independent of \(m(t,0,\lz)\) for all \(t\in\R_+\).

For given \(u(\cdot,x) \in X\) and
\(\lz\in\{\lz\in\C\setminus\{0\}\colon \Im(\lz^2)\leq0\}\), we fix \(t_0>0\)
such that each entry of \(\int_{t_0}^{\infty} |\widetilde{Q}(\tau,0,\lz)|
\id\tau\) is bounded by a constant \(c_{t_0} < 1\) and by the arbitrary choice of \(t_0\),
the non-uniqueness is apparent. For \(\Im(\lz^2)\leq0\), we consider the following
integral equation for \(n(t,0,\lz)\),
\begin{align*}
  n(t,0,\lz) & = e^{2i\lz^2 t} e_2 + \int_{t_0}^{t}
  \begin{pmatrix}
    e^{-2i\lz^2 (t-\tau)} & 0 \\
    0 & 0
  \end{pmatrix} \widetilde{Q}(\tau,0,\lz)\, n(\tau,0,\lz) \id\tau\\
  & \qquad - \int_{t}^{\infty}
  \begin{pmatrix}
    0 & 0 \\
    0 & e^{2i\lz^2 (t-\tau)}
  \end{pmatrix} \widetilde{Q}(\tau,0,\lz)\, n(\tau,0,\lz) \id\tau, \quad
  t\geq t_0.
\end{align*}
Set \(\hat{n}(t,0,\lz)=n(t,0,\lz)e^{-2i\lz^2 t}\),
then the integral equation becomes
\begin{equation}\label{eq:vol2}
  \hat{n}(t,0,\lz) = e_2 + (\mathcal{N}\hat{n})(t,0,\lz), \quad t\geq t_0,
\end{equation}
where \(\mathcal{N}\) is an integral operator defined by
\begin{align*}
  (\mathcal{N}\hat{n})(t,0,\lz)=&\int_{t_0}^{t}
  \begin{pmatrix}
    e^{-4i\lz^2 (t-\tau)} & 0 \\
    0 & 0
  \end{pmatrix} \widetilde{Q}(\tau,0,\lz)\, \hat{n}(\tau,0,\lz) \id\tau\\
  &- \int_{t}^{\infty}
  \begin{pmatrix}
    0 & 0 \\
    0 & 1
  \end{pmatrix} \widetilde{Q}(\tau,0,\lz)\, \hat{n}(\tau,0,\lz) \id\tau,\quad
  \hat{n}(\cdot,0,\lz) \in L^\infty[t_0,\infty).
\end{align*}
By the same argument as for \(\hat{m}(t,0,\lz)\), we have existence of
\(\hat{n}(t,0,\lz)\) for \(t\in (t_0,\infty)\). As \(\Im(\lz^2)\leq0\)
and each entry of \(\widetilde{Q}(\cdot,0,\lz)\) being in \(L^1[t_0,\infty)\),
\(\mathcal{N}\) is a bounded operator from \(L^\infty[t_0,\infty)\) to
\(L^\infty[t_0,\infty)\). Similar to before, put \(\hat{n}_0(t,0,\lz)= e_2\) and define
\(\hat{n}_{j+1}(t,0,\lz)= e_2+ (\mathcal{N}\hat{n}_j)(t,0,\lz)\), inductively.
Then,
\begin{equation*}
  \|(\hat{n}_{j+1}-\hat{n}_j)(\cdot,0,\lz)\|_{L^\infty[t_0,\infty)}
  \leq c_{t_0}^j, \quad j\geq0.
\end{equation*}
Indeed \(\|\hat{n}_1(\cdot,0,\lz)-\hat{n}_0(\cdot,0,\lz)\|_{L^\infty[t_0,\infty)}
\leq c_{t_0}\) and for \(j\geq1\),
\begin{align*}
  \|(\hat{n}_{j+1}-\hat{n}_j)(\cdot,0,\lz)\|_{L^\infty[t_0,\infty)}
  &=\|(\mathcal{N}(\hat{n}_j-\hat{n}_{j-1}))(\cdot,0,\lz)\|_{L^\infty[t_0,\infty)}\\
  & \leq \|(\hat{n}_j-\hat{n}_{j-1})(\cdot,0,\lz)\|_{
  L^\infty[t_0,\infty)} \int_{t_0}^{\infty} |\widetilde{Q}(\tau,0,\lz)|\id\tau\\
  &=c_{t_0} \|(\hat{n}_j-\hat{n}_{j-1})(\cdot,0,\lz)\|_{L^\infty[t_0,\infty)}
\end{align*}
Therefore, \(\hat{n}(t,0,\lz)=\hat{n}_0(t,0,\lz)+\sum_{j=1}^{\infty}
\hat{n}_j(t,0,\lz)-\hat{n}_{j-1}(t,0,\lz)\) converges in
\(L^\infty[t_0,\infty)\) and solves the integral equation \eqref{eq:vol2}.
Writing \(\hat{n}(t,0,\lz)= (\hat{n}_1,\hat{n}_2)^\intercal\), \eqref{eq:vol2}
becomes
\begin{align*}
  \hat{n}_1(t,0,\lz) & = \int_{t_0}^{t} e^{-4i\lz^2 (t-\tau)}
  (\widetilde{Q}_{11}(\tau,0,\lz)\hat{n}_1(\tau,0,\lz) +
  \widetilde{Q}_{12}(\tau,0,\lz)\hat{n}_2(\tau,0,\lz)) \id\tau\\
  \hat{n}_2(t,0,\lz) & =1 - \int_{t}^{\infty}
  \widetilde{Q}_{21}(\tau,0,\lz)\hat{n}_1(\tau,0,\lz) +
  \widetilde{Q}_{22}(\tau,0,\lz)\hat{n}_2(\tau,0,\lz) \id\tau
\end{align*}
As for \(\hat{m}(t,0,\lz)\), we can prove, if \(u(\cdot,0)\in X\), then
\(\hat{n}_1(\cdot,0,\lz)\in H^{1,1}_t[t_0,\infty)\). Therefore,
we consider with Lemma \ref{l:est} the estimate
\begin{align*}
  \|\hat{n}_1(\cdot,0,\lz)\|_{H^{1,1}_t[t_0,\infty)}&\leq
  c\|(\widetilde{Q}_{11}\hat{n}_1)(\cdot,0,\lz)\|_{H^{0,1}_t[t_0,\infty)}+
  c\|(\widetilde{Q}_{12}\hat{n}_2)(\cdot,0,\lz)\|_{H^{0,1}_t[t_0,\infty)}\\
  &\leq\|\hat{n}_1(\cdot,0,\lz)\|_{L^\infty[t_0,\infty)}
  \|\widetilde{Q}_{11}(\cdot,0,\lz)\|_{H^{0,1}_t(\R_+)}\\&\quad
  +\|\hat{n}_2(\cdot,0,\lz)\|_{L^\infty[t_0,\infty)}
  \|\widetilde{Q}_{12}(\cdot,0,\lz)\|_{H^{0,1}_t(\R_+)}.
\intertext{A similar reasoning involving Lemma \ref{l:est} implies
that \(\hat{n}_2(\cdot,0,\lz)-1\in H^{1,1}_t[t_0,\infty)\). We have}
  \|\hat{n}_2(\cdot,0,\lz)-1\|_{H^{1,1}_t[t_0,\infty)}&\leq
  c\|(\widetilde{Q}_{21}\hat{n}_1)(\cdot,0,\lz)\|_{H^{1,1}_t[t_0,\infty)}+
  c\|(\widetilde{Q}_{22}\hat{n}_2)(\cdot,0,\lz)\|_{H^{1,1}_t[t_0,\infty)}\\
  &\leq\|\hat{n}_1(\cdot,0,\lz)\|_{H^{1,1}_t[t_0,\infty)}
  \|\widetilde{Q}_{21}(\cdot,0,\lz)\|_{H^{0,1}_t(\R_+)}\\&\quad
  +\|\hat{n}_2(\cdot,0,\lz)\|_{L^\infty[t_0,\infty)}
  \|\widetilde{Q}_{22}(\cdot,0,\lz)\|_{H^{1,1}_t(\R_+)}.
\end{align*}
Except for
\begin{equation*}
  \|\widetilde{Q}_{22}(\cdot,0,\lz)\|_{H^{1,1}_t(\R_+)}\leq
  \|u(\cdot,0)\|_{L^\infty(\R_+)}\|u(\cdot,0)\|_{H^{1,1}_t(\R_+)},
\end{equation*}
all estimates on the entries of \(\widetilde{Q}(t,0,\lz)\) are already done
in \eqref{eq:wQest}. Therefore, we indeed have that \(\hat{n}(\cdot,0,\lz)-e_2
\in H^{1,1}_t[t_0,\infty)\) if \(u(\cdot,0) \in X\).
We know that \(n(t,0,\lz)\) defined through \(\hat{n}(t,0,\lz)\) solves
the integral equation \eqref{eq:volmh} for \(t\in\R_+\) and we have its existence
in \(t\geq t_0\), it follows that, given \(t_0\), \(n(t,0,\lz)\) can be
uniquely extended to a solution of the \(t\) part of the Lax pair for
\(\Im(\lz^2)<0\).

The linear independence of \(m(t,0,\lz_1)\) and \(n(t,0,\lz_1)\),
\(\lz_1\in\{\lz\in\C\colon \Im(\lz^2)<0\}\), can be shown by
\begin{equation*}
  \lim_{t\to\infty}\det(m(t,0,\lz_1),n(t,0,\lz_1))=1.
\end{equation*}
Since \(V\) has zero trace, we conclude that
\begin{equation*}
  \det(m(t,0,\lz_1),n(t,0,\lz_1))=1, \quad t\geq0.
\end{equation*}
Then, for \(t\geq0\), we can write \(\js_1\) as a linear combination
of \(m(t,0,\lz_1)\) and \(n(t,0,\lz_1)\) such that
\begin{equation*}
  \js_1(t) = c_1 m(t,0,\lz_1) + c_2 n(t,0,\lz_1)
\end{equation*}
for some constants \(c_1\), \(c_2\). If \(c_2=0\), then as \(t\to\infty\),
\begin{equation*}
  \psi_0(t) = c_1 e^{-2i\lz_1^2 t}
  \begin{pmatrix}
    1+r_2(t) \\
    r_3(t)
  \end{pmatrix},\quad r_j\in H^{1,1}_t(\R_+),\quad j=2,3.
\end{equation*}
Hence,
\begin{equation*}
  (P[1])_{12} = \frac{(1+r_2(t)) r_3(t)^\ast}{
  |1+r_2(t)|^2+|r_3(t)|^2}\in H^{1,1}_t(\R_+)
\end{equation*}
As in the argumentation for the B\"acklund matrix being a map from \(u(\cdot,0)\in\LI
\to \LI\ni u_1(\cdot,0)\), the denominator \(|1+r_2(t)|^2+|r_3(t)|^2\) can not be zero,
due to \(m(t,0,\lz)\) being a solution to the spectral problem \(\js_t = (-2i\lz^2\st
+ \widetilde{Q})\js\) and given its asymptotic behavior as \(t\) goes to infinty.
If \(c_2\neq0\), then as \(t\to\infty\),
\begin{equation*}
  \psi_1(t) = c_2 e^{2i\lz_1^2 t}
  \begin{pmatrix}
    r_4(t) \\
    1+r_5(t)
  \end{pmatrix},\quad r_j\in H^{1,1}_t(\R_+),\quad j=4,5.
\end{equation*}
The same reasoning makes sure that the denominator can not be zero and hence,
\begin{equation*}
  (P[1])_{12} = \frac{(1+r_5(t))^\ast r_4(t)}{
  |1+r_4(t)|^2+|r_5(t)|^2}\in H^{1,1}_t(\R_+)
\end{equation*}
Thus,
\begin{equation*}
  u_1=u+(P[1])_{12}\in H^{1,1}_t(\R_+).
\end{equation*}
By the second line of equation \eqref{eq:u1}, it can also be shown that
\((u_1)_x(\cdot,0)\in H^{0,1}_t(\R_+)\) in both cases,
which implies \(u_1(\cdot,0)\in X\). For \(\Im(\lz^2)\geq0\),
the choice of the normalization of \(m(t,0,\lz)\) and \(n(t,0,\lz)\) is reversed.
\end{proof}
\begin{lem}\label{l:diag}
Let \(u(\cdot,0)\in X\), and \(D[1]\) be a dressing matrix constructed by
\(\lz_1=\xi+i\eta\) and \(\js_1=(\mu,\nu)^\intercal\)
evaluated at \(x=0\). Then, \(D[1]\big\rvert_{x=0}\) goes to either
\(\diag(\lz-\lz_1^\ast, \lz-\lz_1)\) or
\(\diag(\lz-\lz_1,\lz-\lz_1^\ast)\) as \(t\to\infty\),
depending on the limit behavior of \(\psi_1\).
\end{lem}
\begin{proof}
At \(t=0\) and \(x=0\), \(\js_1\) is either being produced by
\((1,c)^\intercal\), \(c\in\C\), or \((0,1)^\intercal\).
In the first case, \(\js_1 = c_1 m(t,0,\lz_1)+ c_2 n(t,0,\lz_1)\)
for some constants \(c_1\), \(c_2\), where \(m(t,0,\lz)\) and \(n(t,0,\lz)\)
are the linear independent solutions of the \(t\) part of the Lax system
as constructed in the proof of Proposition \ref{p:bij}.
If \(\js_1\) is proportional to \(m(t,0,\lz_1)\), then necessarily \(c_2=0\).
As a consequence \(\frac{\nu}{\mu} = \frac{m_2(t,0,\lz_1)}{
m_1(t,0,\lz_1)}\to 0\) as \(t\to\infty\) and so
\(D[1]\big\rvert_{x=0}\) goes to \(\diag(\lz-\lz_1,\lz-\lz_1^\ast)\)
as \(t\to\infty\). If \(c_2\neq0\), then, as \(t\to\infty\),
\(\js_1 = c_2 e^{2i\lz_1^2 t}
\begin{pmatrix}
  r_4(t) \\
  1+r_5(t)
\end{pmatrix}\), where \(r_4,r_5\in H^{1,1}_t(\R_+)\) as before. Therefore,
\(\frac{\mu}{\nu} \to 0\) as \(t\to\infty\) and so
\(D[1]\big\rvert_{x=0}\) goes to \(\diag(\lz-\lz_1^\ast,\lz-\lz_1)\) as \(t\to\infty\).
In the second case, we necessarily have \(c_1=0\) and again by
\(n(t,0,\lz_1)\), we have that \(D[1]\big\rvert_{x=0}\) goes to \(\diag(\lz-\lz_1^\ast,
\lz-\lz_1)\) as \(t\to\infty\).
\end{proof} 
\section{Dressing on two half-lines}\label{sc:wl}
In this section, we want to use the introduced Darboux transformation to construct soliton solutions for a model PDE on two half-lines which are connected via
boundary conditions established through the localized B\"acklund transformation.
Based on the ideas of \cite{Zh}, we will incorporate this boundary condition into
the dressing process. In this way, we also show that this method can be,
without further complications, generalized to a simple graph structure and
to a time-dependent gauge transformation.
\subsection{NLS equation with defect conditions}
For the convenience of the reader, we will explicitly state the model
as a generalization of the NLS equation \eqref{eq:NLS} to the NLS equation
on two half-lines
\begin{align}\label{eq:NLShl1}
\begin{split}
  &i u_t + u_{xx}+ 2 |u|^2u = 0,\\
  &u(0,x) = u_0(x)
\end{split}
  \intertext{for \(u(t,x)\colon \R_+ \times \R_+ \mapsto \C\) and initial condition \(u_0(x)\) for \(x\in \R_+\) and}
\begin{split}\label{eq:NLShl2}
  &i \tilde{u}_t + \tilde{u}_{xx} + 2|\tilde{u}|^2\tilde{u} = 0,\\
  &\tilde{u}(0,x) = \tilde{u}_0(x)
\end{split}
\end{align}
for \(\tilde{u}(t,x)\colon \R_+ \times \R_- \mapsto \C\) and initial condition
\(\tilde{u}_0(x)\) for \(x\in \R_-\). In that context, taking for example
\(u(t,0) = \tilde{u}(t,0)\) and \(u_x(t,0) = \tilde{u}_x(t,0)\) as boundary conditions,
the two half-lines are connected such that there is no reflection and trivial
transmission and by redefining the initial condition accordingly, we end up with
the NLS equation as in \eqref{eq:NLS}. However, the model we are interested in
arises with so-called defect conditions
\begin{align}\label{eq:def}
\begin{split}
  &(\tilde{u} - u)_x = i \alpha (\tilde{u} - u) \pm \Omega(\tilde{u} + u),\\
  &(\tilde{u} - u)_t = -\alpha (\tilde{u} - u)_x \pm i \Omega (\tilde{u} + u)_x
  +i(\tilde{u} - u)(|u|^2+|\tilde{u}|^2) \qquad
\end{split}
\end{align}
at \(x=0\). In particular, we have \(\Omega=\sqrt{\beta^2 - |\tilde{u}-u|^2}\)
and defect parameter \(\alpha, \beta \in \R\). It can be shown that the defect
conditions \eqref{eq:def} are equivalent to the general localized B\"acklund
transformation \eqref{eq:locdef} with \(u\), \(\tilde{u}\), \(\alpha\), \(\beta\)
and the respective sign as we have already matched. Since we want to incorporate
the defect into the dressing method, we define the localized defect matrix
corresponding to seed solutions \(u[0](t,x)\) and \(\tilde{u}[0](t,x)\)
to the NLS equations on the respective half-line \(\R_+\times\R_-\)
and \(\R_+\times\R_+\) and two spectral parameter \(\alpha\in\R\) and
\(\beta\in\R\setminus\{0\}\) such that
\begin{equation}\label{eq:G0}
G_0(t,0,\lz)= 2\lz\I +
\begin{pmatrix}
  \alpha\pm i\sqrt{\beta^2-|\tilde{u}[0]-u[0]|^2} & -i(\tilde{u}[0]-u[0]) \\
  -i(\tilde{u}[0]-u[0])^\ast & \alpha\mp i\sqrt{\beta^2-|\tilde{u}[0]-u[0]|^2}
\end{pmatrix}
\end{equation}
satisfies \eqref{eq:locdef}. Here, the sign in the \((11)\)-entry and
accordingly the \((22)\)-entry can still be chosen freely. The goal
in this section is to successfully apply the properties we have developed
in order to dress the solutions \(u[0](t,x)\) and \(\tilde{u}[0](t,x)\) in
such a way that we are able to find a localized defect matrix preserving the
boundary constraint. In that regard, we talk about a similar form if
the boundary constraint for the dressed solutions holds with the same
spectral parameter \(\alpha\) and \(\beta\) as well as the same sign
in front of the root of the \((11)\)-entry. It is an important
step in dressing the boundary to handle the sign in the entries
of the matrix for the dressed localized defect matrix. Especially,
we will use Proposition \ref{p:bij} and Lemma \ref{l:diag} to ensure that
the signs match, when the solution is in a particular function space.

\begin{remark}
The connection of the defect conditions \eqref{eq:def} to the boundary
constraint \eqref{eq:locdef} has been discussed, among other publications,
in \cite{Ca2} and \cite{CoZa}. Therein, the authors additionally prove the
existence of an infinite set of modified conservation laws, which means that
the defect conditions are indeed integrable boundary conditions.
\end{remark}
The defect conditions being integrable also establishes the possibility to apply
the Darboux transformation to this model in order to produce soliton solutions.
\subsection{Dressing the boundary}
The results of this section are inspired by \cite{Zh}. However, due to the
differences in the model, we want to apply the dressing the boundary to,
the implementation of the method differs from the original approach.
In particular, the pairing of zeros \(\lz_1\) and \(-\lz_1\) in order
to respect the relation of \(V(t,0,\lz)\) and \(V(t,0,-\lz)\) is not
required in the presented model. There will be, nonetheless, a connection
of the zeros \(\lz_1\) introduced for the positive half-line and the ones
\(\tilde{\lz}_1\) introduced for the negative half-line. With that in mind,
we present the main result of this paper.
\begin{prop}\label{p:Nsol}
Consider solutions \(u[0]\) and \(\tilde{u}[0]\) to the NLS equation
\eqref{eq:NLShl1} and \eqref{eq:NLShl2}, which at \(x=0\) are in the
function space \(X\). In addition, let \(u[0]\) and \(\tilde{u}[0]\)
be subject to the defect condition \eqref{eq:def} with parameter
\(\alpha\in\R\) and \(\beta\in\R\setminus\{0\}\). Further, take
\(N+1\) solutions \(\js_j\), \(j=0,\dots,N\), of the undressed Lax
system corresponding to \(u[0]\) for \(\lz=\lz_0=-\frac{\alpha+i\beta}{2}\)
and distinct \(\lz =\lz_j\in \C\setminus\bigl(\R\cup\{\lz_0,\lz_0^\ast\}
\bigr)\), \(j=1,\dots,N\). Constructing \(G_0\) of localized defect form
as in \eqref{eq:G0} with \(u[0]\), \(\tilde{u}[0]\), \(\alpha\), \(\beta\) and
chosen sign, we assume that there exist paired solutions \(\widetilde{\js}_j\) of
the undressed Lax system corresponding to \(\tilde{u}[0]\) for \(\lz= \lz_j\),
\(j=1,\dots,N\), satisfying
\begin{equation}\label{eq:gpsi}
\widetilde{\js}_j\big\rvert_{x=0} =
G_0(t,0,\lz_j) \js_j\big\rvert_{x=0},\quad j=1,\dots,N.
\end{equation}
Then, two \(N\)-fold Darboux transformations \(D[N]\), \(\widetilde{D}[N]\)
using the corresponding solutions lead to solutions \(u[N]\) and \(\tilde{u}[N]\) to the NLS equations \eqref{eq:NLShl1} and \eqref{eq:NLShl2}.
In particular, the defect conditions are preserved.
\end{prop}
To this end, we shall show that the functions \(u[N]\) and \(\tilde{u}[N]\)
(a) satisfy the NLS equation on the respective half-line,
(b) are in their Lax systems subject to the boundary constraint with
an, for the time being, unspecified matrix \(G_N\),
and further, that (c) \(G_N\) is of a similar localized defect form as \(G_0\).
\begin{proof}
(a) The \(N\)-fold Darboux transformations \(D[N]\), \(\widetilde{D}[N]\)
construct, as presented in Section \ref{sc:darb}, solutions \(u[N]\),
\(\tilde{u}[N]\) from seed solutions \(u[0]\), \(\tilde{u}[0]\),
which satisfy the same partial differential equations. Therefore,
having \(N\) linearly independent solutions is enough to ensure that
the transformed solutions satisfy the respective NLS equation.
In that regard, as already mentioned the linear independence of \(\js_j\)
and \(\widetilde{\js}_j\) is implied by choosing distinct \(\lz_j\),
\(j=1,\dots,N\).

For (b) and (c) the existence of \(\js_0\) and a kernel vector, we call
\(\upsilon_0\), corresponding to the defect matrix \(G_0\) is crucial.
Given the generality of the statement, we have either linear
independence of \(\upsilon_0\) in terms of \(\{\js_0,\jss_0\}\) or linear dependence,
where \(\jss_0=\sigma_2 \js_0^\ast\), \(\sigma_2=
\begin{pmatrix}
  0 & -i \\
  i & 0
\end{pmatrix}\) is an orthogonal vector of \(\js_0\). Moreover,
we have that for a matrix \(G_N\), of the form \(G_N = \lz \I + G^{(0)}\),
the equality
\begin{align*}
  \widetilde{D}[N] G_0 &= G_N D[N],
\intertext{where \(G_0\) is divided by \(2\) (to be similar to
a dressing matrix), is sufficient for \(G_N\) to satisfy}
  (G_N)_t&=\widetilde{V}[N]G_N-G_N V[N]
\end{align*}
at \(x=0\), except for the zeros of the \(t\) part of the Lax systems \(\lz=\lz_j\), \(j=1,\dots,N\).
To show that, we multiply the first equation, evaluated \(x=0\), with
\((D[N])^{-1}\) and differentiate the resulting equation to obtain
\begin{align*}
  (G_N)_t&=(\widetilde{D}[N] G_0 (D[N])^{-1})_t\\
  &=\widetilde{D}_t[N] G_0 (D[N])^{-1}+\widetilde{D}[N] (G_0)_t (D[N])^{-1}+
  \widetilde{D}[N] G_0 ((D[N])^{-1})_t.
\end{align*}
Whereas, the first two summands can be simplified using \eqref{eq:D1} and
\eqref{eq:locdef} such that
\begin{align*}
  \widetilde{D}_t[N] G_0+\widetilde{D}[N] (G_0)_t & =
  (\widetilde{V}[N]\widetilde{D}[N]-\widetilde{D}[N]\widetilde{V}[0])G_0+
  \widetilde{D}[N](\widetilde{V}[0]G_0-G_0V[0]) \\
  & =\widetilde{V}[N]\widetilde{D}[N] G_0 - \widetilde{D}[N]G_0V[0].
\intertext{In addition, it can be shown with \eqref{eq:D1} that for the third
summand}
  ((D[N])^{-1})_t & = -(D[N])^{-1} D_t[N] (D[N])^{-1} \\
  & =-(D[N])^{-1}   V[N] + V[0](D[N])^{-1}
\intertext{holds. Put together and notice that the expressions
\(\widetilde{D}[N] G_0 (D[N])^{-1}\) are in fact again \(G_N\), we obtain}
  (G_N)_t&= \widetilde{V}[N]\widetilde{D}[N] G_0(D[N])^{-1}
  -\widetilde{D}[N] G_0 (D[N])^{-1} V[N]\\
  &=\widetilde{V}[N]G_N-G_N V[N].
\end{align*}
Analogously, the \(x\) part of the defect constraint is implied.
Thereby, we know that if \(G_N\) is of similar form as \(G_0\),
see \eqref{eq:G0}, with \(\tilde{u}[0]-u[0]\) replaced by \(\tilde{u}[N]-u[N]\),
then \(u[N]\) and \(\tilde{u}[N]\) satisfy the defect conditions. However,
let us first prove that we can construct a matrix \(G_N\) of first order in
\(\lz\) such that \(\widetilde{D}[N] G_0 = G_N D[N]\) at \(x=0\), as above.
\begin{figure}
\centering
\begin{tikzpicture}
  \matrix (m) [matrix of math nodes,row sep=3em,column sep=4em,minimum width=2em]
  {
     \js_0 \text{ solves }
     \begin{cases}
       \js_x=U[0]\js\\
       \js_t=V[0]\js
     \end{cases} &
     \js'_0=D[N]\js_0 \text{ solves }
     \begin{cases}
       \js_x=U[N]\js\\
       \js_t=V[N]\js
     \end{cases}\\
     \widetilde{\js}_0=G_0\js_0 \text{ solves }
     \begin{cases}
       \js_x=\widetilde{U}[0]\js\\
       \js_t=\widetilde{V}[0]\js
     \end{cases} &
     \widetilde{\js}'_0=\widetilde{D}[N]\widetilde{\js}_0
     \text{ solves }
     \begin{cases}
       \js_x=\widetilde{U}[N]\js\\
       \js_t=\widetilde{V}[N]\js
     \end{cases}\\};
  \path[-stealth]
    (m-1-1) edge node [left] {$G_0$} (m-2-1)
            edge node [below] {$D[N]$} (m-1-2)
    (m-2-1) edge node [below] {$\widetilde{D}[N]$} (m-2-2)
    (m-1-2) edge[dashed] node [right] {$G_N$} (m-2-2);
\end{tikzpicture}
  \caption{Properties of \(\js_0\) at \(\lz=\lz_0\) and \(x=0\)
  if \(G_0\js_0\neq0\).}\label{f:psi0}
\end{figure}

(b) As mentioned already, the defect matrix \(G_0\) admits, see Remark
\ref{r:defdar}, a kernel vector \(\upsilon_0\) at \(\lz_0\),
since \(\Im(\lz_0)\neq0\). Hence,
\begin{equation*}
  G_0(t,0,\lz_0)\upsilon_0=0.
\end{equation*}
If this vector \(\upsilon_0\) differs from a linear combination of \(\{\js_0,\jss_0\}\),
we have that the diagram of Figure~\ref{f:psi0} holds. Thus,
as \(\js_0\) is linearly independent of \(\js_1,\dots,\js_N\),
\(D[N]\) is invertible at \(\lz=\lz_0\). Therefore, we can
transform \(\js'_0\), the solution to the Lax system
corresponding to \(u[N]\) at \(\lz=\lz_0\),
to \(\bjs'_0 = \widetilde{D}[N]G_0(D[N])^{-1}\js'_0\).
In turn, this vector \(\bjs'_0\) is a solution to the Lax system
corresponding to \(\tilde{u}[N]\) at \(\lz=\lz_0\). Consequently,
the matrix, we call \(G_{N_1}\), given by the product
\(\widetilde{D}[N]G_0 (D[N])^{-1}\), satisfies the equations
\begin{align*}
  ((G_{N_1})_x&-\widetilde{U}[N]G_{N_1}+G_{N_1} U[N])\js'_0 =0,\\
  ((G_{N_1})_t&-\widetilde{V}[N]G_{N_1}+G_{N_1} V[N])\js'_0 =0
\end{align*}
at \(\lz=\lz_0\) and \(x=0\). Then, we have with the equivalence above at \(x=0\)
the following
\begin{align}\label{eq:GNpro1}
\begin{aligned}
  &\widetilde{D}[N]G_0\psi_0 = G_{N_1} D[N]\psi_0\neq 0,&&\lz=\lz_0,\\
  &\widetilde{D}[N]G_0\varphi_0 = G_{N_1} D[N]\varphi_0\neq 0,&&\lz=\lz_0^\ast.
\end{aligned}
\end{align}
It is reasonable to assume that \(G_{N_1}\) is a polynomial matrix of order \(1\),
due to the product \(\widetilde{D}[N]G_0 (D[N])^{-1}\). Indeed, the dressing
matrices can be written as \(\lz^N(\I +\mathcal{O}(\frac{1}{\lz}))\)
and therefore,
\begin{equation*}
  \widetilde{D}[N]G_0 (D[N])^{-1}=(\I +\mathcal{O}(\frac{1}{\lz}))G_0
  (\I +\mathcal{O}(\frac{1}{\lz}))=\lz \I + \widetilde{G}^{(0)} +
  \mathcal{O}(\frac{1}{\lz}).
\end{equation*}
Whereas the term \(\mathcal{O}(\frac{1}{\lz})\) on the right hand side needs
to be identically zero, since we have a product of polynomials.
Further, evaluating the determinant of \(G_{N_1}\) at \(\lz=\lz_0\) and
\(\lz = \lz_0^\ast\), we obtain \(\det(G_{N_1})=0\). This is implying
that there is a kernel vector \(\tilde{\upsilon}'_0\) such that
\(G_{N_1}(t,0,\lz_0)\tilde{\upsilon}'_0=0\). Constructing a one-fold dressing matrix with
\(\tilde{\upsilon}'_0\) at \(\lz = \lz_0\), we obtain
\begin{equation}\label{eq:GNdar1}
  G_{N_1}=(\lz-\lz_0^\ast)\I+(\lz_0^\ast-\lz_0)\widetilde{P}'_0,\quad
  \widetilde{P}'_0=\frac{\tilde{\upsilon}'_0(\tilde{\upsilon}'_0)^\dagger}{
  (\tilde{\upsilon}'_0)^\dagger \tilde{\upsilon}'_0},
\end{equation}
where \(G_{N_1}\) satisfies the property \eqref{eq:GNpro1}.

On the other hand, if w.l.o.g.\@ \(\upsilon_0=\psi_0\), we construct \(G_{N_2}\) with a
different vector. Remember that \(\js_0\) is linearly independent of
\(\js_1,\dots,\js_N\). Define a new vector
\begin{equation*}
  \upsilon'_0 = D[N](t,0,\lz_0)\upsilon_0.
\end{equation*}
In this case, the strategy is to construct a one-fold
dressing matrix with the defined vector \(\upsilon'_0\) at \(\lz = \lz_0\) by
\begin{equation}\label{eq:GNdar2}
  G_{N_2}=(\lz-\lz_0^\ast)\I+(\lz_0^\ast-\lz_0)P'_0,\quad
  P'_0=\frac{\upsilon'_0(\upsilon'_0)^\dagger}{(\upsilon'_0)^\dagger \upsilon'_0},
\end{equation}
such that \(G_{N_2}(t,0,\lz_0)\upsilon'_0=0\). This results in the property
\begin{align}\label{eq:GNpro2}
\begin{aligned}
  &\widetilde{D}[N]G_0\psi_0 = G_{N_2} D[N]\psi_0 = 0,&&\lz=\lz_0,\\
  &\widetilde{D}[N]G_0\varphi_0 = G_{N_2} D[N]\varphi_0= 0,&&\lz=\lz_0^\ast
\end{aligned}
\end{align}
at \(x=0\).

Constructing \(G_N\) as in one of the two cases \(G_{N_1}\) or \(G_{N_2}\) will give
us commutating matrices at the point \(x=0\) of the defect conditions.
\begin{figure}
\centering
\begin{tikzpicture}
  \matrix (m) [matrix of math nodes,row sep=3em,column sep=4em,minimum width=2em]
  {
     u[0] & u[N] \\
     \tilde{u}[0] & \tilde{u}[N] \\};
  \path[-stealth]
    (m-1-1) edge node [left] {$G_0$} (m-2-1)
            edge node [below] {$D[N]$} (m-1-2)
    (m-2-1) edge node [below] {$\widetilde{D}[N]$} (m-2-2)
    (m-1-2) edge node [right] {$G_N$} (m-2-2);
  \begin{scope}[xshift=15em]
  \matrix (m) [matrix of math nodes,row sep=3em,column sep=4em,minimum width=2em]
  {
     u[0] & u[N] \\
     \tilde{u}[0] & \tilde{u}[N] \\};
  \path[-stealth]
    (m-2-1) edge node [left] {$G_0^{-1}$} (m-1-1)
    (m-1-1) edge node [below] {$D[N]$} (m-1-2)
    (m-2-1) edge node [below] {$\widetilde{D}[N]$} (m-2-2)
    (m-2-2) edge node [right] {$G_N^{-1}$} (m-1-2);
  \end{scope}
\end{tikzpicture}
  \caption{Permutability of defect matrices at \(x=0\).}\label{f:com}
\end{figure}
In particular, we can now show that
\begin{equation}\label{eq:lr}
  (\widetilde{D}[N] G_0)\big\rvert_{x=0} = (G_N D[N])\big\rvert_{x=0}.
\end{equation}
To prove \eqref{eq:lr}, we write each side as a matrix polynomial,
by dividing \(G_0\) by \(2\) and denoting the left and right
hand side respectively as \(L(\lz)\) and \(R(\lz)\), we obtain
in both cases, \(N=N_1\) or \(N=N_2\), the following
\begin{align*}
L(\lz)&= \widetilde{D}[N] G_0 = \lz^{N+1} L_{N+1}+ \lz^N L_N + \dots +\lz L_1 + L_0,\\
R(\lz)&= G_N D[N] =  \lz^{N+1} R_{N+1}+ \lz^N R_N + \dots +\lz R_1 + R_0.
\end{align*}
Since \(L_{N+1}=\I=R_{N+1}\), only \(L_N\), \(R_N\), \dots, \(L_1\), \(R_1\),
\(L_0\) and \(R_0\) need to be determined. In that regard, we consider the zeros
and associated kernel vectors of \(L(\lz)\) and \(R(\lz)\).
By construction of the dressing matrices \(D[N]\), \(\widetilde{D}[N]\),
we have that \(D[N](t,x,\lz_j)\js_j=0\) and \(\widetilde{D}[N]
(t,x,\lz_j)\bjs_j=0\), \(j=1,\dots,N\),
which we will combine with the assumed relation between \(\js_j\)
and \(\bjs_j\). Thus, for the \(N\) linearly independent
\(\js_1,\dots,\js_N\), we have
\begin{align*}
  &L(\lz)\big\rvert_{\lz=\lz_j}\js_j=0,
  && R(\lz)\big\rvert_{\lz=\lz_j}\js_j=0,
\intertext{\(j=1,\dots,N\), whereby these equalities hold for \(x=0\).
Here, the symmetry of the Lax pair provides another vector \(\jss_j
= \sigma_2 \js_j^\ast\), which is orthogonal to \(\js_j\). Analogously, let
\(\widetilde{\jss}_j=\sigma_2 \bjs_j^\ast\) and it
follows that}
  &L(\lz)\big\rvert_{\lz=\lz_j^\ast} \jss_j=0,
  &&R(\lz)\big\rvert_{\lz=\lz_j^\ast} \jss_j=0
\end{align*}
for \(j=1,\dots,N\) and \(x=0\). For a defect matrix of order one, this is
not enough to ensure equality in \eqref{eq:lr}. However,
we constructed \(G_N\) in a way such that there is an additional vector pair
for which the two sides are equal. For \(N=N_2\), it should be noted that even if
the kernel vector is a linear combination of \(\js_0\) and \(\jss_0\), it is
possible to repeat the following steps, but the notation becomes unhandy without
giving more insight. Hence for \(N=N_1\) and \(N=N_2\), we consider
\begin{align*}
  &L(\lz)\big\rvert_{\lz=\lz_0}\js_0=
  R(\lz)\big\rvert_{\lz=\lz_0}\js_0,
\intertext{whereby this equality is either nonzero for \(N=N_1\) or
zero for \(N=N_2\) and holds for \(x=0\). As before, the symmetry of the
Lax pair provides another vector \(\jss_0 = \sigma_2 \js_0^\ast\),
which is orthogonal to \(\js_0\) and for \(x=0\), it satisfies}
  &L(\lz)\big\rvert_{\lz=\lz_0^\ast} \jss_0=
  R(\lz)\big\rvert_{\lz=\lz_0^\ast} \jss_0.
\end{align*}
This \emph{additional} pair of vectors determines \(L(\lz)-R(\lz)=
C(\lz)=\lz^N C_N +\dots +\lz C_1 + C_0\).
Together with the zeros and associated kernel vectors of the Darboux matrices
\(D[N]\), \(\widetilde{D}[N]\),  it can be written as a set of algebraic equations
\begin{equation*}
\begin{array}{rcr}%
  \Bigl(\lz_0^N C_N +\dots + \lz_0 C_1+C_0\Bigr) \js_0=0,& &
  \Bigl((\lz_0^\ast)^N C_N +\dots + \lz_0^\ast C_1+C_0\Bigr) \jss_0=0,\\
  \vdots\;\quad&&\vdots\;\quad\\
  \Bigl(\lz_N^N C_N +\dots + \lz_N C_1+C_0\Bigr) \js_N=0,& &
  \Bigl((\lz_N^\ast)^N C_N +\dots + \lz_N^\ast C_1+C_0\Bigr) \jss_N=0.
\end{array}
\end{equation*}
In matrix form, we have
\begin{equation*}
  (C_N,\cdots,C_0)
  \begin{pmatrix}
    \lz_0^N \js_0 & (\lz_0^\ast)^N \jss_0 & \cdots & \lz_N^N \js_N
    & (\lz_N^\ast)^N \jss_N \\
    \vdots & \vdots & \vdots & \vdots & \vdots \\
    \js_0 & \jss_0 & \cdots & \js_N & \jss_N
  \end{pmatrix}=0.
\end{equation*}
The \((2N+2)\times(2N+2)\) matrix  filled with \(\{\js_0,\jss_0,\dots,
\js_N,\jss_N\}\) is invertible. If the determinant was zero,
we could find coefficients in \(\C\) such that a linear combination
of \(\{\psi_0,\varphi_0,\dots,\psi_N,\varphi_N\}\) would be zero,
which is a contradiction to their linear independence. If, for \(N=N_2\),
\(\upsilon_0\) was a linear combination of \(\js_0\) and \(\jss_0\),
the matrix would still be invertible with this linear combination
and its orthogonal in the first and second column, respectively. Thereby,
\(L(\lz)=R(\lz)\) holds in both cases \(N=N_1\) as well as \(N=N_2\),
which, in turn, implies that (b) is satisfied.

By (b), we have matrices \(G_{N_1}\) and \(G_{N_2}\) which satisfy the boundary
constraint and are of the form of a dressing matrix at \(x=0\). Further,
the equality \eqref{eq:lr} ensures that \(\tilde{\upsilon}'_0\) is also
for \(G_{N_1}\) equal to \(D[N](t,0,\lz_0)\upsilon_0\).
It is an important fact that in both cases the kernel vector \(\upsilon_0\)
of \(G_0(t,0,\lz_0)\) takes the role of the kernel vector \(\tilde{\upsilon}'_0\)
and \(\upsilon'_0\) of respectively \(G_{N_1}\) and \(G_{N_2}\). Only then,
we can prove that the explicit forms of \(G_{N_1}\) and \(G_{N_2}\)
are consistent with \(G_0\) through Proposition \ref{p:bij} and Lemma \ref{l:diag}.
In general, we can not think of them being the same, since we assumed for
\(N=N_2\) that w.l.o.g.\@ \(\upsilon_0=\psi_0\).
Nevertheless, the equality \eqref{eq:lr} already provides the localized defect form
of \(G_{N_1}\) and \(G_{N_2}\) which was commented on in Remark \ref{r:defdar}
except for the sign in front of the \((11)\)-entry.

(c) With the given information we are able to find the localized defect form
of \(G_N = \lz\I + \widetilde{G}^{(0)}\), where the proof is similar in both cases
\(G_{N_1}\) and \(G_{N_2}\). That is, from the off-diagonal of \(L_N=R_N\),
it can be seen that \(\widetilde{G}^{(0)}_{12}\) and \(\widetilde{G}^{(0)}_{21}\)
can respectively be written as \(-i(\tilde{u}[N]-u[N])/2\) and
\(-i(\tilde{u}[N]-u[N])^\ast/2\) at \(x=0\), which gives
\begin{equation}\label{eq:GNv1}
  G_N(t,0,\lz)=\lz \I +\frac{1}{2}
  \begin{pmatrix}
    2 \widetilde{G}^{(0)}_{11}&-i(\tilde{u}[N]-u[N])\\
    -i(\tilde{u}[N]-u[N])^\ast &2 \widetilde{G}^{(0)}_{22}
  \end{pmatrix}.
\end{equation}
We can compute the determinant of \(G_N\), where we use the property of
determinants of Darboux transformations, so that at \(x=0\),
\begin{align*}
  \det(G_N) &= \det(\widetilde{D}[N])\det(G_0)\det((D[N])^{-1})=\det(G_0)\\
  &=(\lz-\lz_0)(\lz-\lz_0^\ast)
  = \lz^2 +\alpha \lz + \frac{\alpha^2+\beta^2}{4}
\end{align*}
and in particular, the determinant is independent of \(t\) and \(x\).
Comparing with \eqref{eq:GNv1}, we obtain
\begin{align*}
  \alpha & = \widetilde{G}^{(0)}_{11}+\widetilde{G}^{(0)}_{22},\\
  \alpha^2+\beta^2 & = 4\widetilde{G}^{(0)}_{11}\cdot
  \widetilde{G}^{(0)}_{22}+|\tilde{u}[N]-u[N]|^2.
\end{align*}
Solving for \(\widetilde{G}^{(0)}_{11}\) and \(\widetilde{G}^{(0)}_{22}\) at \(x=0\),
we have
\begin{align*}
  2\widetilde{G}^{(0)}_{11} & = \alpha \pm i\sqrt{\beta^2-|\tilde{u}[N]-u[N]|^2},\\
  2\widetilde{G}^{(0)}_{22} & = \alpha \mp i\sqrt{\beta^2-|\tilde{u}[N]-u[N]|^2}.
\end{align*}
However, at this point the signs of \(\widetilde{G}^{(0)}_{11}\) and
\(\widetilde{G}^{(0)}_{22}\) are not necessarily the same as the signs
of the defect matrix \(G_0\). In that regard, we know that from solutions
\(u[0]\), \(\tilde{u}[0]\) to the defect conditions with a selected sign,
we can construct solutions \(u[N]\), \(\tilde{u}[N]\) which satisfy the defect
conditions with either the plus or the minus sign. A particular case can be determined
for which we are able to prove that the signs stay the same, ultimately restricting
the solution space.

Assuming that we have
\begin{equation}\label{eq:utinf}
  u[0](\cdot,0), \tilde{u}[0](\cdot,0)\in X,
\end{equation}
then \(u[N](\cdot,0), \tilde{u}[N](\cdot,0) \in X\) by Proposition \ref{p:bij},
since \(D[N]\), \(\widetilde{D}[N]\) are \(N\) transformations of the form
\(\mathcal{B}^t_{\lz_j}(\js_j)\) for \(j=1,\dots,N\).
In that class of solutions, we can identify the signs for matrices \(G_0\) and
\(G_N\) of localized defect form through the kernel vectors respective to their
form as Darboux transformation. We know that in both cases \(\upsilon_0\) is the
kernel vector for \(G_0\) at \(\lz=\lz_0\) and by construction, we have that
\(\omega_0=D[N](t,0,\lz_0)\upsilon_0\) is the kernel vector of \(G_N\) at
\(\lz=\lz_0\) and \(x=0\). On the other hand, as \(t\) goes to infinity \(G_0\)
becomes a diagonal matrix and as a consequence, the limit behaviors of \(\js_j\),
\(\bjs_j\) are the same for \(j=1,\dots,N\), since they are connected through \(G_0\).
Consequently, the dressing matrices \(\widetilde{D}[N]\) and \(D[N]\) have the same
distribution of \(\lz-\lz_j\) and \(\lz-\lz_j^\ast\) in their
diagonal form as \(t\to\infty\). Thus,
\begin{equation}\label{eq:Ginf}
  \lim_{t\to\infty}G_N = \lim_{t\to\infty} \widetilde{D}[N] G_0 (D[N])^{-1}
  = \lim_{t\to\infty} G_0.
\end{equation}
Also the vectors \(\upsilon_0\) and \(\omega_0\), respectively the kernel vectors of
\(G_0\) and \(G_N\) at \(\lz=\lz_0\), admit the same limit behavior as \(t\to\infty\),
since they are connected by \(D[N]\) which admits a diagonal structure
as \(t\to\infty\). Starting with a plus (minus) sign in the \((11)\)-entry of
the defect matrix \(G_0\), we can then conclude by the limit behavior of \(\upsilon_0\)
and \(\omega_0\) as well as \eqref{eq:Ginf} that the sign in the \((11)\)-entry of the
defect matrix \(G_N\) needs to plus (minus). Therefore, they satisfy similar
defect conditions on the spectral side with \(V[0]\), \(\widetilde{V}[0]\) replaced
by \(V[N]\) and \(\widetilde{V}[N]\), which gives the result in the solution class \(X\).

With \(G_N\) of similar localized defect form and
satisfying \((G_N)_t= \widetilde{V}[N]G_N-G_N V[N]\) at \(x=0\),
we can conclude that the defect conditions are preserved for \(u[N]\)
and \(\tilde{u}[N]\).
\end{proof}

With Proposition \ref{p:Nsol}, we proved that \emph{dressing the boundary}
can be applied to the NLS equation on a simple star-graph with a non trivial
boundary condition. Therefore, extending the method, presented in \cite{Zh}, to more
than one half-line and also considering time-dependent boundary matrices.
Thereby, we have given a way to construct \(N\)-soliton solutions for
particular seed solutions \(u[0]\) and \(\tilde{u}[0]\). It should be
mentioned that, since the seed solutions are a part of the localized defect
matrix \(G_0\), their influence on the construction of
\(\tilde{u}[N]\) is decisive. In fact, it is a priori not clear,
whether there exists a solution \(\bjs_1\) at \(\lz=\lz_1\) to the
undressed Lax system of \eqref{eq:NLShl2} satisfying \eqref{eq:gpsi}.

The important feature of the proof is that the localized defect matrix \(G_0\)
is interchangeable with the Darboux matrices \(D[N]\) and \(\widetilde{D}[N]\)
in the sense of Figure \ref{f:com}. In turn, this is realized with the
transformation of the localized defect matrix \(G_0\) to a
Darboux transformation and vice versa the Darboux matrix \(G_N\)
to a localized defect matrix, which has been mentioned in
Remark \ref{r:defdar}. Apparently, this transformation
process is not necessary until applying a time-dependent boundary matrix.
\begin{remark}
In \cite{CoZa} similar results of a two-soliton solution subject to the defect conditions have been
presented without utilizing the spectral side of the model. With the background of a B{\"a}cklund
transformation, the solution was \emph{assumed} to be an individual soliton on each side of the defect,
it was checked with an algebra program that these functions indeed solve the defect conditions,
however only with \(\alpha=0\).
\end{remark}%
In the following section, we want to give an application of Proposition \ref{p:Nsol}. 
\section{Dressing soliton solutions}\label{sc:vis}
\subsection{N-soliton solutions}
Consider the zero seed solutions \(u[0]=\tilde{u}[0]=0\). Particularly,
\(u[0](\cdot,0), \tilde{u}[0](\cdot,0)\in X\) and \(u[0](t,\cdot)\in H^{1,1}(\R_+)\),
\(\tilde{u}[0](t,\cdot)\in H^{1,1}(\R_-)\) and \(\C\setminus \R
\ni\lambda_0=-\frac{\alpha\pm i\beta}{2}\). Hence,
\begin{equation*}
G_0(t,0,\lz)= 2\lz\I +
\begin{pmatrix}
  \alpha\pm i\beta & 0 \\
  0 & \alpha\mp i\beta
\end{pmatrix}.
\end{equation*}
So for solutions \(\js_j\) to the Lax pair corresponding to \(u[0]\)
at \(\lz=\lz_j\in \C\setminus \bigl(\R\cup\{\lz_0,\lz_0^\ast\}\bigr)\),
\(j=0,\dots,N\), we have
\begin{equation}\label{eq:Nsolj}
  \js_j =
  \begin{pmatrix}
    \mu_j \\
    \nu_j
  \end{pmatrix} =
  e^{(-i\lz_j x - 2 i\lz_j^2 t)\st}
  \begin{pmatrix}
    u_j \\
    v_j
  \end{pmatrix}
\end{equation}
with \((u_j,v_j)\in \C^2\) and since the relation
\begin{equation*}
  \bjs_j\big\rvert_{x=0}=
  G_0(t,0,\lz_j)\js_j\big\rvert_{x=0}
\end{equation*}
should hold for \(j=1,\dots,N\) and solutions defined by
\(\bjs_j= e^{(-i\lz_j x - 2 i\lz_j^2 t)\st} (\tilde{u}_j,\tilde{v}_j)^\intercal\),
\((\tilde{u}_j,\tilde{v}_j)\in \C^2\), of the Lax system corresponding to
\(\tilde{u}[0]\) at \(\lz=\lz_j\), we obtain the following relation for the
spectral parameter \(\tilde{u}_j\), \(\tilde{v}_j\), \(u_j\) and \(v_j\),
\begin{equation*}
  \frac{\tilde{u}_j}{\tilde{v}_j} = \frac{2\lz_j+\alpha\pm i\beta}{
  2\lz_j+\alpha \mp i\beta}\frac{u_j}{v_j},\quad j=1,\dots,N.
\end{equation*}
This is enough to apply Proposition \ref{p:Nsol}.
Note that changing the sign of \(\beta\) is the same as changing the sign in the
defect conditions. We also know that the \(N\)-soliton solution constructed with
Proposition \ref{p:Nsol} satisfy \(u[N](\cdot,0), \tilde{u}[N](\cdot,0)\in X\)
and \(u[N](t,\cdot)\in H^{1,1}(\R_+)\), \(\tilde{u}[N](t,\cdot)\in H^{1,1}(\R_-)\),
due to Proposition \ref{p:bij} and Proposition 4.7 in \cite{DZ}, which can easily
be extended to Darboux transformations where \(\lz_1\) has a real part.
Moreover, similar analysis holds true for \(u[0](t,\cdot)=0\) in \(H^{1,1}(\R)\),
then \(u[N](t,\cdot) \in H^{1,1}(\R)\).
As in the proof of Proposition \ref{p:Nsol}, we can use this fact to make
sure that, after finding the defect form \(B_N(t,x,\lz)\) for \(x\in \R\)
of the localized defect matrix \(G_N(t,0,\lz)\), the sign in front of the
root in the \((11)\)-entry is consistent with the sign of the defect form
\(B_0(t,x,\lz)=G_0(t,0,\lz)\) for \(x\in \R\) of the localized defect matrix
\(G_0(t,0,\lz)\). Ultimately, we can use this extension to show that each
soliton interacts with the defect individually.

Taking the same Darboux transformations, however, applying them to zero seed
solutions \(u[0]\) and \(\tilde{u}[0]\) on the whole line \(x\in\R\), we obtain
two \(N\)-soliton solutions \(u_N(t,x)\) and \(\tilde{u}_N(t,x)\) for the NLS
equation for \(x\in\R\). Suppose their corresponding
solutions to the Lax system are related by the dressing transformation
\begin{equation*}
  \bjs(t,x,\lz)=B_N(t,x,\lz)\js(t,x,\lz).
\end{equation*}
Then the matrix \(B_N(t,x,\lz)\) solves the system \eqref{eq:defect}. As explained before,
assuming this matrix is linear in \(\lz\), it can only be of the form described
in Proposition \ref{p:def}, which means there exist real parameter \(\delta\), \(\gamma
\in \R\) and a \(\pm\) sign to be determined such that
\begin{equation*}
  B_N(t,x,\lz)=2\lz\I +
  \begin{pmatrix}
    \delta\pm i\sqrt{\gamma^2-|\tilde{u}_N - u_N|^2} & -i(\tilde{u}_N - u_N) \\
    -i(\tilde{u}_N - u_N)^\ast & \delta\mp i\sqrt{\gamma^2-|\tilde{u}_N - u_N|^2}
  \end{pmatrix}.
\end{equation*}
At their respective half-line, the full line solutions \(u_N(t,x)\) and
\(\tilde{u}_N(t,x)\) can be reduced to their half-line counterpart \(u[N](t,x)\),
\(\tilde{u}[N](t,x)\). Hence,
\begin{equation*}
  B_N(t,0,\lz)=2\lz\I +
  \begin{pmatrix}
    \delta\pm i\sqrt{\gamma^2-|\tilde{u}[N] - u[N]|^2} & -i(\tilde{u}[N] - u[N]) \\
    -i(\tilde{u}[N] - u[N])^\ast & \delta\mp i\sqrt{\gamma^2-|\tilde{u}[N] - u[N]|^2}
  \end{pmatrix}.
\end{equation*}
However, at \(x=0\), we know that the two solutions \(u[N](t,x)\) and \(\tilde{u}[N](t,x)\)
can be connected with the defect matrix \(G_N(t,0,\lz)\) used in the proof of
Proposition \ref{p:Nsol}, i.e.\@ \(B_N(t,0,\lz)=2G_N(t,0,\lz)\). Therefore,
we can deduce that \(\delta = \alpha\), \(\gamma^2 = \beta^2\). This means
that the matrix \(G_N(t,0,\lz)\), constructed
in the proof in order to show that the boundary condition is preserved, has in
fact a continuation \(B_N(t,x,\lz)\) for \(x\in\R\). Due to \(u_N(t,\cdot),
\tilde{u}_N(t,\cdot) \in H^{1,1}(\R)\), we have that as \(x\) goes to plus or
minus infinity:
\begin{equation*}
  \lim_{|x|\to\infty} B_N(t,x,\lz) = 2\lz\I +
  \begin{pmatrix}
    \alpha\pm i|\beta| & 0 \\
    0 & \alpha\mp i|\beta|
  \end{pmatrix}.
\end{equation*}
As before, we can make out the exact sign through the kernel vectors. For the \(N\)-soliton
solution, we have that the kernel vector for \(G_0\) can easily be connected to a solution
of the Lax system. Therefore, we take \(\js_0\) as in equation \eqref{eq:Nsolj},
where \(\lz_0=-\frac{\alpha\pm i\beta}{2}\), \(u_0\in \C\setminus\{0\}\) arbitrary and
\(v_0=0\). Here, the \(\pm\) sign in \(\lz_0\) is the same as in the localized defect matrix
\(G_0\). Continuing \(G_0\) to a defect matrix \(B_0(t,x,\lz)\) for
both zero seed solutions on the full line, we see that the kernel vector \(\js_0\)
carries the information of the signs as \(|x|\to\infty\). Now, we know that the Darboux
transformed kernel vector \(\bjs_0=D[N](t,0,\lz_0)\js_0\) is the kernel vector
for \(G_N(t,0,\lz)\) and hence for \(B_N(t,0,\lz_0)\). However, since in this case the
kernel vector is at the same time a solution to the \(x\) part of the Lax system,
we obtain at \(\lz=\lz_0\) the following
\begin{equation*}
  (B_N)_x \bjs_0=\widetilde{U} B_N \bjs_0 - B_N U \bjs_0=
  \widetilde{U} B_N \bjs_0- B_N (\bjs_0)_x.
\end{equation*}
Thus, \((B_N\bjs_0)_x=\widetilde{U} B_N \bjs_0\) at \(\lz=\lz_0\) and every \(x\in\R\).
In turn, this implies, given \(B_N(t,0,\lz_0)\bjs_0=0\), that \(B_N(t,x,\lz_0)\bjs_0=0\)
for every \(x\in \R\). Then, notice that \(\bjs_0\) has the same asymptotic behavior as
\(\js_0\), since the dressing matrix goes to a diagonal matrix for \(|x|\to\infty\).
As a consequence, the signs in the entries of \(\lim_{x\to\pm\infty}B_N(t,x,\lz)\) are
completely determined by the signs of the limits from the defect matrix \(B_0(t,x,\lz)\).
Which amounts in the problem presented to
\begin{equation*}
  B_{\infty}(\lz)=\lim_{|x|\to\infty} B_N(t,x,\lz) = 2\lz\I +
  \begin{pmatrix}
    \alpha\pm i\beta & 0 \\
    0 & \alpha\mp i\beta
  \end{pmatrix}.
\end{equation*}
Knowing that, we see that the Jost solutions have relations induced by the B{\"a}cklund
transformation and the same normalization factor \(B_{\infty}^{-1}(\lz)\),
\begin{equation*}
  \bjs_\pm(t,x,\lz)=B_N(t,x,\lz)\js_\pm(t,x,\lz)B_{\infty}^{-1}(\lz).
\end{equation*}
In turn, this relation implies the following relation for the corresponding
scattering matrices:
\begin{equation}\label{eq:Arel}
  \widetilde{A}(\lz) = B_{\infty}(\lz) A(\lz)B_{\infty}^{-1}(\lz),\quad \lz\in \R.
\end{equation}
\begin{cor}\label{c:def}
Let \(u(t,x)\) and \(\tilde{u}(t,x)\) be two pure N-soliton solutions of the NLS equation
on \(\R\) constructed by the corresponding vectors used in Proposition \ref{p:Nsol}
and let their restrictions to respectively the positive and negative half-line be
subject to the defect conditions \eqref{eq:def} at \(x=0\). Then, it follows
for \(\lz_j\in\C_+\) that solitons are transmitted through the defect
independently of one another, i.e.\@ for all \(j=1,\dots,N\) the following holds
\begin{align*}
\tilde{x}_j-x_j&=\frac{1}{2\eta_j}\log\Bigl(\Bigl|
\frac{2\lz_j+\alpha\mp i\beta}{2\lz_j+\alpha\pm i\beta}\Bigr|\Bigr),\\
\tilde{\varphi}_j-\varphi_j&= \arg\Bigl(
\frac{2\lz_j+\alpha\mp i\beta}{2\lz_j+\alpha\pm i\beta}\Bigr).
\end{align*}
\end{cor}
\begin{proof}
By the analysis above, we know in this case that
\begin{equation*}
  B_{\infty}(\lz)=\lim_{|x|\to\infty}B_N(t,x,\lz)=(2\lz+\alpha)\I\pm i\beta\st,
\end{equation*}
where the sign in front of \(\beta\) matches the sign of the defect.
The relation of the Jost solutions gives
\begin{equation}\label{eq:jostbeta}
\begin{aligned}
  \bjs_-^{(1)}&=B_N(t,x,\lz) \js_-^{(1)} (2\lz+\alpha\pm i\beta)^{-1}, \\
  \bjs_+^{(2)}&=B_N(t,x,\lz) \js_+^{(2)} (2\lz+\alpha\mp i\beta)^{-1}.
\end{aligned}
\end{equation}
We can deduce using \eqref{eq:jostbeta} and \eqref{eq:bj} that
\begin{align*}
  \bjs_-^{(1)}(t,x,\lz_j) &= \frac{b_j}{2\lz_j+\alpha\pm i\beta}
  B_N(t,x,\lz_j)\js_+^{(2)}(t,x,\lz_j)
  \intertext{and with \eqref{eq:jostbeta} and the corresponding weight relations
  to \eqref{eq:bj} for the extended solution \(\tilde{u}(t,x)\), we obtain}
  &=\frac{2\lz_j+\alpha\mp i\beta}{2\lz_j+\alpha\pm i\beta}
  \frac{b_j}{\tilde{b}_j} \bjs_-^{(1)}(t,x,\lz_j).
\end{align*}
Therefore, the constants \(\tilde{b}_j\) and \(b_j\) can be related by
\begin{equation}\label{eq:bjrel}
  \frac{\tilde{b}_j}{b_j}=\frac{2\lz_j+\alpha\mp i\beta}{2\lz_j+\alpha\pm i\beta}.
\end{equation}
Moreover, the relation \eqref{eq:Arel} for the scattering matrices implies
\begin{align}\label{eq:a22rel}
  \tilde{a}_{22}(\lz) &= a_{22}(\lz),\\\notag
  \tilde{a}_{12}(\lz) &= \frac{2\lz+\alpha\pm i\beta}{2\lz+\alpha \mp i\beta}
  a_{12}(\lz).
\end{align}
These two relation \eqref{eq:bjrel} and \eqref{eq:a22rel} can be combined
to relate the weights \(\tilde{C}_j\) and \(C_j\) in the following way
\begin{equation*}
  \frac{\tilde{C}_j}{C_j}=\frac{\tilde{b}_j}{b_j}\frac{a'_{22}(\lz_j)}{
  \tilde{a}'_{22}(\lz_j)}=\frac{2\lz_j+\alpha\mp i\beta}{2\lz_j+\alpha\pm i\beta},
\end{equation*}
from where we can see the influence on the \(N\)-soliton solution. Thereby,
writing the norming constants as
\begin{equation*}
  C_j = 2\eta_j e^{2\eta_j x_j+ i\varphi_j},\quad
  \tilde{C}_j = 2\eta_j e^{2\eta_j \tilde{x}_j+ i\tilde{\varphi}_j}
\end{equation*}
for \(j=1,\dots,N\) as motivated for the one-soliton solution in Section \ref{sc:NLS},
we obtain for the spatial shift \(\tilde{x}_j-x_j\) and the phase shift
\(\tilde{\varphi}_j-\varphi_j\) the following
\begin{align*}
\tilde{x}_j-x_j&=\frac{1}{2\eta_j}\log\Bigl(\Bigl|
\frac{2\lz_j+\alpha\mp i\beta}{2\lz_j+\alpha\pm i\beta}\Bigr|\Bigr),\\
\tilde{\varphi}_j-\varphi_j&= \arg\Bigl(
\frac{2\lz_j+\alpha\mp i\beta}{2\lz_j+\alpha\pm i\beta}\Bigr),
\end{align*}
which implies that solitons experience independently of one another.
\end{proof}
\begin{remark}
Another way of proving Corollary \ref{c:def} is to use Theorem 1.22 of \cite{GHZ},
where it is shown that the scattering data is, after successive iteration of the
Darboux transformation, in each step, which we indicate by \([j]\), given by
\begin{equation*}
  \begin{aligned}
  &a_{11}^{[j]}(\lz) =\frac{\lz-\lz_j}{\lz-\lz_j^\ast} a_{11}^{[j-1]}(\lz),
  &&\tilde{a}_{11}^{[j]}(\lz) =\frac{\lz-\lz_j}{\lz-\lz_j^\ast}\tilde{a}_{11}^{[j-1]}(\lz),
  &&\lz \in \C_+\cup \R,\\
  &a_{21}^{[j]}(\lz) =a_{21}^{[j-1]}(\lz),
  &&\tilde{a}_{21}^{[j]}(\lz) =\tilde{a}_{21}^{[j-1]}(\lz),
  &&\lz \in \R,\\
  &C_k^{[j]} = \frac{\lz_k-\lz_j^\ast}{\lz_k-\lz_j}C_k^{[j-1]},
  &&\tilde{C}_k^{[j]} = \frac{\lz_k-\lz_j^\ast}{\lz_k-\lz_j}\tilde{C}_k^{[j-1]},
  &&k = 1,\dots,j-1,\\
  &C_j^{[j]} = \frac{\lz_j-\lz_j^\ast}{-\frac{v_j^\ast}{u_j^\ast}a_{11}^{[j-1]}(\lz_j)},
  &&\tilde{C}_j^{[j]} = \frac{\lz_j-\lz_j^\ast}{-\frac{\tilde{v}_j^\ast}{\tilde{u}_j^\ast}
  \tilde{a}_{11}^{[j-1]}(\lz_j)}.
  \end{aligned}
\end{equation*}
Therefore, given that \(a_{11}^{[0]}(\lambda)=1\) and \(a_{12}^{[0]}(\lambda)=0\)
for the zero seed solution \(u[0](t,x)=0\), we have that
\begin{equation*}
  \frac{\tilde{C}_j^{[N]}}{C_j^{[N]}} = \frac{\tilde{u}_j^\ast}{\tilde{v}_j^\ast}
  \frac{v_j^\ast}{u_j^\ast}=\frac{2\lz_j+\alpha\mp i\beta}{2\lz_j+\alpha\pm i\beta},
  \quad j=1,\dots,N.
\end{equation*}
\end{remark}
The complex conjugation of the quotients \(\frac{\tilde{u}_j^\ast}{\tilde{v}_j^\ast}\)
and \(\frac{v_j^\ast}{u_j^\ast}\) is due to the fact that in the referenced book, the
dressing is done with the Jost functions which effectively go to \(e_1\) and \(e_2\)
as \(x\) respectively goes to \(-\infty\) and \(\infty\). In order for the \(\js_j\)
to comply with that requirement, we need to change their asymptotic behavior while
making sure that the \(\lz_j\) are chosen correctly. In summary, the theorem affirms
Corollary \ref{c:def}.
\begin{figure}
  \centering
  \includegraphics{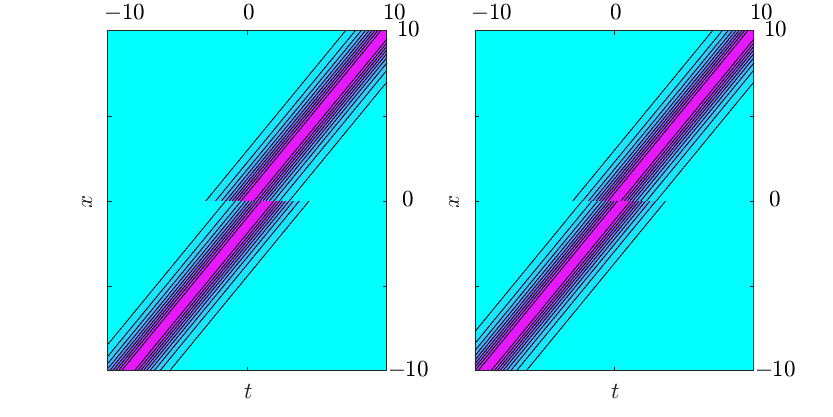}
  \caption{Contour plot of a one-soliton solution satisfying the NLS equation
  on each half-line and the defect conditions \eqref{eq:def} with defect
  parameter \(\alpha=0\) and \(\beta=1\) (left) as well as \(\beta=3\) (right).
  }\label{fig:1sol}
\end{figure}
\begin{figure}
  \centering
  \includegraphics{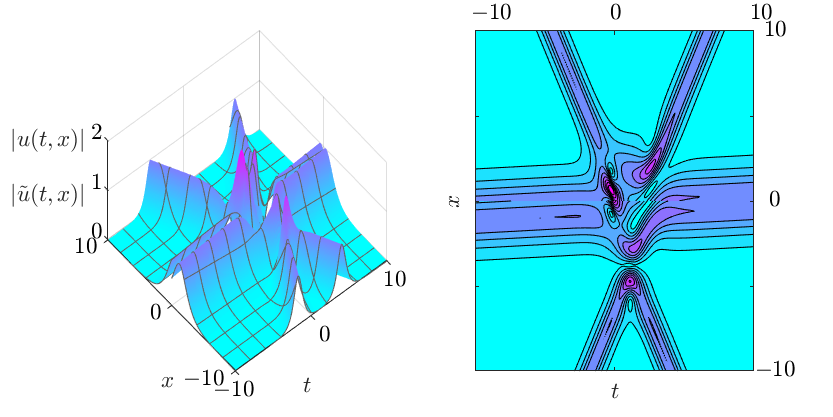}
  \caption{Plot of a three-soliton solution (left) and its contour (right)
  satisfying the NLS equation on each half-line and the defect conditions
  \eqref{eq:def} with defect parameter \(\alpha=0\) and \(\beta=1\).}\label{fig:3sol}
\end{figure}
\begin{figure}
  \centering
  \includegraphics{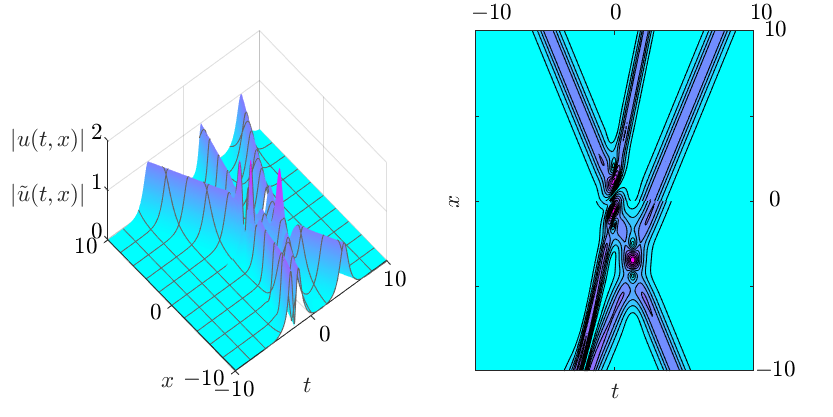}
  \caption{Plot of a three-soliton solution (left) and its contour (right)
  satisfying the NLS equation on each half-line and the defect conditions
  \eqref{eq:def} with defect parameter \(\alpha=0\) and \(\beta=1\).}\label{fig:3sol2}
\end{figure}

The expression \(\frac{2\lz_j+\alpha\mp i\beta}{2\lz_j+\alpha\pm i\beta}\) lets us
state some facts about the behavior of the spatial and phase shift of the \(N\)-soliton
after interacting with the defect. Letting \(\beta\) go to zero, the quotient goes
to \(1\), which indicates the discontinuity at \(x=0\) disappears, suggesting that
\(\alpha\) by itself can not maintain it. Whereas letting \(|\beta|\) go to infinity,
the quotient goes to \(-1\), which means no considerable spatial shift as
\(\tilde{x}_j-x_j\) goes to zero and essentially a shape inversion as
\(\tilde{\varphi}_j-\varphi_j\) goes to \(\pi\) for all \(j=1,\dots,N\).
However, if we take \(\beta\in \R\setminus\{0\}\) and let \(|\alpha|\) go to infinity,
the effect of the discontinuity also disappears, i.e.\@ \(\tilde{x}_j-x_j\) and
\(\tilde{\varphi}_j-\varphi_j\) both go to zero for all \(j=1,\dots,N\). Hence,
the second defect parameter may be understood as a means to smooth out the
discontinuity in the presence of the defect condition \((\beta\neq0)\). Therefore,
the discontinuity reaches its full potential, when \(\alpha=0\).

In this regard, we plotted the absolute value of the one-soliton solutions, \(u[1]\)
and \(\tilde{u}[1]\) satisfying the NLS equation on \(\R_+\) and \(\R_-\) and being
subject to the defect condition, in Figure \ref{fig:1sol} and thereby showing the effect
of an increasing defect parameter \(\beta\). In Figure \ref{fig:3sol}, we plot the
absolute value of a three-soliton solution \(u[3]\)
and \(\tilde{u}[3]\) satisfying the presented model with defect
parameter \(\alpha=0\) and \(\beta=1\) and also its contour. All of the three solitons
have the same amplitude, two of them have opposite velocity and the velocity
of the third soliton is chosen to be slow in order to show the discontinuity. Again,
one can observe a smoothing effect when choosing either \(\alpha\) not equal to zero
or \(\beta\) large enough. Similarly, a three-soliton solution, where the slow moving
soliton is replaced by a fast moving soliton with the same amplitude, is shown in
Figure \ref{fig:3sol2}. Conceptually, higher order soliton solutions could be computed
and plotted.

The authors of \cite{CoZa} have been investigating the construction of soliton solutions
by confirming through direct calculation that the one- and two-soliton solutions
satisfy the defect conditions. For the convenience of the reader, we give the
connection to the notation therein for the one-soliton solution.
\begin{remark}
To translate the expression into the notation used in \cite{CoZa}, first off we
need to take \(\beta=0\) and additionally \(\Omega=\sqrt{\alpha^2 - |\tilde{u}-u|^2}\).
Then, for the one-soliton solution consider \(\frac{v_1}{u_1}=1\), \(a = 2\eta\),
\(c=-2\xi\), \(p=e^{-2\eta \tilde{x}_1}\) and finally \(q=e^{-i \tilde{\varphi}_1}\)
to recover the same result.
\end{remark}
\subsection{Destructive soliton solution}
There is also a particular solution known for which we can avoid restricting the solution
space to \(X\). Beginning again with zero seed solutions \(u[0]=0\), \(\tilde{u}[0]=0\),
\(\alpha\in \R\), \(\beta\in \R\setminus\{0\}\) and a choice in the sign in front of the
root in the \((11)\)-entry of the localized defect matrix \(G_0(t,0,\lz)\), we can construct
\begin{equation*}
  G_0(t,0,\lz) = 2\diag(\lz-\lz_0,\lz-\lz_0^\ast),
\end{equation*}
where \(\lz_0 = -\frac{\alpha\pm i\beta}{2}\). Taking the same spectral parameter \(\lz_0\)
to construct a dressing matrix on one side, we take \(\R_-\), of the defect together
with a solution \(\bjs_0=(\mu,\nu)^\intercal\) to the undressed Lax system \eqref{eq:jost}
at \(\lz=\lz_0\) corresponding to \(\tilde{u}[0]\), we obtain
\begin{equation*}
\widetilde{D}[1]=(\lz-\lz_0^\ast)\I+(\lz_0^\ast-\lz_0) P[1],\quad
P[1] = \frac{\bjs_0 \bjs_0^\dagger}{\bjs_0^\dagger\bjs_0}.
\end{equation*}
For the sign in front of the root in the \((11)\)-entry of the localized defect
matrix \(G_0(t,0,\lz)\) to be plus or minus, the solution \(\bjs_0=(\mu,\nu)^\intercal\)
respectively needs to have the limit value \(e_1\) or \(e_2\)
as \(x \to -\infty\). On the other half-line \(\R_+\),
we assume that the solution stays the same \(u[1]=0\).
Then, constructing \(G_1(t,0,\lz)\) as dressing matrix with \(\lz_0\)
and corresponding vector \(\bjs_0=(\mu,\nu)^\intercal\) such that at \(x=0\)
this vector is the kernel vector of \(G_1\), we have everything we need in order
to prove that the boundary constraint is preserved. Remark~\ref{r:defdar}
suggests, that \(G_1\) can be written in localized defect form connecting
\(u[1]\) and \(\tilde{u}[1]\), in other words
\begin{equation*}
\begin{aligned}
  (G_1)_x &=\widetilde{U}[1] G_1- G_1 U[1]=\widetilde{U}[1] G_1- G_1 (-i\lz\st),\\
  (G_1)_t &=\widetilde{V}[1] G_1- G_1 V[1]=\widetilde{V}[1] G_1- G_1 (-2i\lz^2\st),
\end{aligned}
\end{equation*}
which follows directly from the property of \(\widetilde{D}[1]\). Further,
the assumed limit behavior of the solution \(\bjs_0\) makes sure
that---after extending \(G_1\) to a defect matrix on \(x\in\R_-\)---the
form is similar to the form of \(G_0\) in terms of the sign in front of the root.
\begin{remark}
It is possible to switch the roles of \(u[1]\) and \(\tilde{u}[1]\),
keeping the zero solution on \(\R_-\) and dressing at \(\lz_0\) for \(\R_+\)
with predetermined limit behavior of the corresponding solution \(\js_0\) of
the Lax system corresponding to \(u[0]\). However, in that case \(G_1\) needs
to be constructed as \((D[1])^{-1}\) in order for the localized defect form
to connect \(u[1]\) and \(\tilde{u}[1]\), since \(D[1]\) has different partial
differential equations than \(\widetilde{D}[1]\).
\end{remark}
In Figure \ref{fig:dsol}, we plotted two examples of one-soliton solutions
interacting destructively with the defect condition. As mentioned before,
the amplitude and velocity of the soliton is prescribed by the strength of
the defect parameter \(\alpha\) and \(\beta\). The spatial and phase shift
however can be chosen arbitrarily. The idea of these solutions emerged as a
special case, when working with \(\alpha=0\), in order to find nonlinear
counterparts of bound states which are solutions to the linear, potential-free,
Schr\"odinger equation with a defect in \cite{CoZa}. Here, this idea takes the
form of a boundary-bound soliton solution and a one-soliton solution on one of the
half-lines.
\begin{remark}
Taking \(\alpha=0\), the destructive soliton solution is in fact a boundary-bound
soliton solution, which especially is not covered by Proposition \ref{p:Nsol}.
\end{remark}
\begin{figure}
  \centering
  \includegraphics{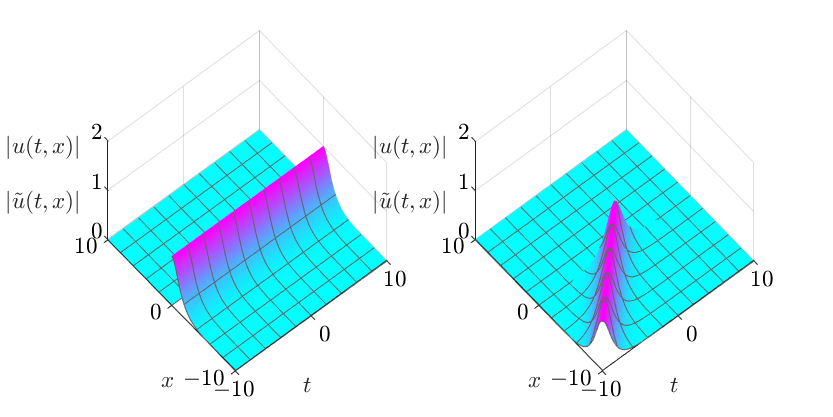}
  \caption{Plot of a boundary-bound and a one-soliton solution satisfying the NLS
  equation on each half-line and interacting destructively with the defect conditions
  \eqref{eq:def} with defect parameter \(\alpha=0\) (left) as well as \(\alpha=0.5\)
  (right) and \(\beta=1\).
  }\label{fig:dsol}
\end{figure} 
\section*{Conclusion}
The defect conditions are subject to some interesting properties as classical systems are,
due to the fact that they stem from a localized B\"acklund transformation. In combination
with Darboux transformations, we provide a direct method in order to compute exact
solutions of the focusing NLS equation on two half-lines connected via the defect
conditions. By carefully reviewing the properties we need, we give a reduction of
the class of solutions which is needed to determine the localized defect matrix
for the Darboux transformation of the solutions. Thereby, we introduce the method
of dressing the boundary to a system consisting of two half-lines connected through
a boundary condition induced by the localized B\"acklund transformation. Hence,
not only readjusting the method to encompass a simple star-graph, but also putting
forward an application on time-dependent gauge transformations, ultimately generalizing
the method presented in \cite{Zh} in two ways.

To discuss soliton behavior, we provided additionally the proof of the conjecture
formulated in \cite{CoZa}: In the model of the NLS equations on two half-lines
connected via defect conditions, an arbitrary number of solitons are transmitted
through the defect independently of one another. Through this means, we
simultaneously made it more relatable to the works where the mirror-image technique
was employed \cite{BiHw}. Thereby, the question arises whether it is possible
to use the mirror-image technique in the model presented to arrive the same results.

The analysis we carried out for a Darboux transformation with respect to \(t\) is
an analogous result to the one with respect to \(x\) given in \cite{DZ}. That being said,
it is also possible to apply the same analysis at an arbitrary point \(x_f\in \R\).
Since the defect conditions can simultaneously be shifted, the results we have are easily
applicable for a defect condition at an arbitrary point \(x_f\in\R\) connecting
two semi-infinite sets in \(\R\).

It is reasonable to assume that the method of dressing the boundary can be applied
to a wide range of systems on which integrable boundary structures exist.
The closest application would be to combine the results in this paper with the results
of \cite{Zh} to obtain soliton solutions for the new boundary conditions \cite{Za} for
the NLS equation on the half-line. Nonetheless, it could also be used to extend
the results \cite{ZCZ} on the sine-Gordon equation with integrable boundary to
include time-dependent transformations on the half-line. Other systems with integrable
boundaries are presented in \cite{Ca2, CoZa} and it can be investigated if the
method extends to these systems. 


\begin{thebibliography}{10}

\bibitem{A}
M.~J. Ablowitz, B.~Prinari, and A.~D. Trubatch.
\newblock {\em Discrete and continuous nonlinear Schr{\"o}dinger systems},
  volume 302.
\newblock Cambridge University Press, 2004.

\bibitem{BiHw}
G.~Biondini and G.~Hwang.
\newblock Solitons, boundary value problems and a nonlinear method of images.
\newblock {\em Journal of Physics A: Mathematical and Theoretical}, 42(20),
  2009.

\bibitem{Ca2}
V.~Caudrelier.
\newblock On a systematic approach to defects in classical integrable field
  theories.
\newblock {\em International Journal of Geometric Methods in Modern Physics},
  5(07):1085--1108, 2008.

\bibitem{Ca}
V.~Caudrelier.
\newblock On the inverse scattering method for integrable {P}{D}{E}s on a star
  graph.
\newblock {\em Communications in Mathematical Physics}, 338, 09 2014.

\bibitem{CoZa}
E.~Corrigan and C.~Zambon.
\newblock Jump-defects in the nonlinear {Schr{\"o}dinger} model and other
  non-relativistic field theories.
\newblock {\em Nonlinearity}, 19(6):1447, 2006.

\bibitem{DZ}
P.~Deift and J.~Park.
\newblock Long-time asymptotics for solutions of the {N}{L}{S} equation with a
  delta potential and even initial data.
\newblock {\em International Mathematics Research Notices},
  2011(24):5505--5624, 2011.

\bibitem{FoIt}
A.~S. Fokas and A.~R. Its.
\newblock The linearization of the initial-boundary value problem of the
  nonlinear {Schr{\"o}dinger} equation.
\newblock {\em SIAM Journal on Mathematical Analysis}, 27(3):738--764, 1996.

\bibitem{Fo}
A.~S. Fokas, A.~R. Its, and L.-Y. Sung.
\newblock The nonlinear {Schr{\"o}dinger} equation on the half-line.
\newblock {\em Nonlinearity}, 18:1771, 05 2005.

\bibitem{Fu}
A.~S. Fokas and B.~Pelloni.
\newblock {\em Unified transform for boundary value problems: Applications and
  advances}.
\newblock SIAM, 2014.

\bibitem{GHZ}
C.~Gu, A.~Hu, and Z.~Zhou.
\newblock {\em Darboux transformations in integrable systems: theory and their
  applications to geometry}, volume~26.
\newblock Springer Science \& Business Media, 2006.

\bibitem{MaSa}
V.~B. Matveev and M.~A. Salle.
\newblock {\em Darboux Transformations and Solitons}.
\newblock Springer-Verlag, Berlin, 1991.

\bibitem{Za}
C.~Zambon.
\newblock The classical nonlinear {Schr{\"o}dinger} model with a new integrable
  boundary.
\newblock {\em Journal of High Energy Physics}, 2014(8):36, 2014.

\bibitem{Zh}
C.~Zhang.
\newblock Dressing the boundary: On soliton solutions of the nonlinear
  {Schr{\"o}dinger} equation on the half-line.
\newblock {\em Studies in Applied Mathematics}, 2018.

\bibitem{ZCZ}
C.~Zhang, Q.~Cheng, and D.-J. Zhang.
\newblock Soliton solutions of the sine-{G}ordon equation on the half line.
\newblock {\em Applied Mathematics Letters}, 86:64--69, 2018.

\end{thebibliography}
\end{document}